\documentclass{llncs}

\usepackage[square,numbers]{natbib}
\usepackage{amsmath,amscd,amssymb,amsgen,amsfonts,amsbsy}
\usepackage{mathrsfs}
\usepackage[utf8x]{inputenc}
\usepackage{color}
\usepackage{textcomp}
\usepackage{float}
\usepackage{latexsym,graphicx}
\usepackage[caption=false]{subfig}
\usepackage{tikz}
\usetikzlibrary{automata,arrows,positioning,calc}
\usepackage{url}
\usepackage{algorithmic}
\usepackage[ruled,lined]{algorithm2e}

\newcommand{\prob}[1]{\mathsf{Pr}\left( #1 \right)}

\newcommand{\remove}[1]{}

\newcommand{\comments}[1]{}

\newcommand{\expect}[1]{\mathsf{E}\left({#1}\right)}
\newcommand{\given}{\; \big\vert \;} 

\begin{document}

\title{ A Hidden Markov Restless Multi-armed Bandit Model for Playout Recommendation Systems}

%
%
\author{Rahul Meshram\inst{1} \and Aditya Gopalan\inst{2} 
\and D. Manjunath\inst{1}
  \thanks{
    Part of this paper have appeared in  COMSNETS-2017 
    conference, \cite{Meshram17}.} }
%
%
    
\institute{Elecl. Engg. Dept., Indian Institute of Technology Bombay, 
Mumbai-400076, India \\
\and
ECE Dept., Indian Institute of Science, Bangalore-560012, India}

\maketitle              

\begin{abstract}
  We consider a restless multi-armed bandit (RMAB) in which 
  there are two types of arms, say A and B. Each arm
  can be in one of two states, say $0$ or $1.$ 
   Playing a type A arm brings
  it to state $0$ with probability one and not playing it induces
  state transitions with arm-dependent probabilities. 
  Whereas playing a type B arm leads it to state $1$ with probability 
  $1$ and not playing it gets  state that dependent on transition probabilities of 
  arm. 
  Further, play of an arm
  generates a unit reward with a probability that depends on the state
  of the arm. The belief about the state of the arm can be calculated
  using a Bayesian update after every play. This RMAB has been
  designed for use in recommendation systems where the user's
  preferences depend on the history of recommendations. This RMAB can
  also be used in applications like creating of playlists or placement
  of advertisements. 
  In this paper we formulate the long term reward maximization problem as infinite
  horizon discounted reward and average reward  problem. We analyse the RMAB  
  by first studying discounted reward scenario. We show that it is 
  Whittle-indexable and then obtain a closed form
  expression for the Whittle index for each arm calculated from the
  belief about its state and the parameters that describe the arm. 
  We next analyse the average reward problem using vanishing discounted 
  approach and derive the closed form expression for Whittle index.
  For a RMAB to be useful in practice, we need to be able to learn
  the parameters of the arms. We present an algorithm derived from
  Thompson sampling scheme, that learns the parameters of the arms and
  also illustrate its performance numerically.

\keywords{restless multi-armed bandit; recommendation systems; POMDP; automated 
playlist creation systems; learning }
\end{abstract}

\section{Introduction}

Recommendations systems are used in almost all forms of modern media
like YouTube and other video streaming services, Spotify and other
music streaming services to create playlists. Playlists are also
created on personal devices like digital music players. Highly
personalised playlists are now being created using a variety of
information that is mined from behavior history and social networking
sites. A search for patents on playlist creation yields more than
handful of items indicating a strong commercial interest in the search
for good algorithm. The research literature though is scant.

A possible approach to create a playlist is to treat it as a `matrix
completion' problem, (e.g. \cite{Candes10}, and choose a set of items
from the completed matrix for which the user has a high preference.)  A
second approach would be to treat playlist creation as a
recommendation system and generate a sequence of recommendations for
the user.  This is the view that we take in this paper but with the
key of the user's interest in items being influenced by immediate
behavioral history. To the best of our knowledge such a system has not
been considered in the literature. Specifically, we assume that user
will like different items to be repeated at different rates and there
will be some randomness in the preferences. The playlist creation
system that we describe in this paper generates a dynamic list by
taking a binary feedback from the user after an item has been played.
Specifically, we allow different items to have different `return
times' in that some items may be played out more frequently than
others.

We  assume that  the items are of two different
categories---`normal' or type-A item and `viral' or type-B item.
For a normal item, the user preference goes from high low immediately
after playing it and rises slowly after not playing it.
The opposite is true for a `viral' item, i.e., the
 preference goes to high immediately after play and decreases when 
 it is not played.
In this model, the recommendation system observes the user preference 
for the item only after playing it and  not observed  for other items. 
This feedback  is accounted in subsequent plays of items.
Since the user preferences are not observed for other items, system 
maintains the belief about the preference of the item. It is updated 
based on the action and outcome.
The goal of system is to maximize a long term reward function. 
Such system can be modeled a  restless multi-armed bandit (RMAB) with 
hidden states. 

Typically, the users
may have  different preferences towards different
items and also have different repetition rates. It may not be known at the 
beginning. 
Thus, the  learning of associated state transition models and click
through probabilities is required.

We next discuss the related literature on RMAB and learning in 
multi-armed bandits.
\subsection{Related Literature}

Typical recommendation system models based on multi-armed bandits take
the form of contextual bandits, e.g., \cite{Langford07,Li10}. A key
feature in these models is that the user interests are assumed
independent of the immediate recommendation history, i.e., 
the reward model is a
static. We introduce the feature of making it dependent on the immediate 
recommendation history. Other models that address user reaction in
recommendations use a finite sequence of past user responses as a
basis for deciding the current recommendation, e.g., \cite{Hariri12};
these are numerical studies and no provable properties are derived. 

The classical stochastic multi-armed bandit problem for recommendation systems (RS) 
studied in \cite{AuerCF02,Bubeck12,Lai85}, where RS  chooses the items from given 
set of items at each time step. Play of an item yields a random 
reward that is drawn from probability distribution associated with item 
and it is unknown. There goal is to maximize the expected cumulative reward.  
This model do not have states associated with each item and rewards are drawn 
independently at each time step. It is studied as online learning problem.
The performance is measured via regret, it is defined as difference between reward 
obtained using optimal strategy and reward obtained from strategy that is used for learning. A variant of stochastic multi-armed bandit considered in \cite{Caron12}.
These are solved using efficient algorithms  based on upper confidence bound (UCB).
Further, it is shown 
that expected regret scales logarithmically with time steps.  
Recently, Thompson sampling (TS) based Bayesian  algorithm considered for 
stochastic multi-armed in \cite{AgrawalG,Chapelle11,Gopalan15,GopManMan14:thompson}.
  In \cite{Chapelle11}, authors have empirically 
illustrated the performance of Thompson 
sampling algorithm and observed that it performs better than UCB. 
TS algorithm is  analysed in \cite{AgrawalG,Gopalan15} and shown that the regret scales logarithmically  with number of time steps.
 
Another stochastic bandit, a restless multi-armed bandit first studied in 
seminal work of \cite{Whittle88}, where each
  arm has states associated with it and states evolve according to a Markov chain 
  and that evolution is action dependent.  
 Further, author proposed the heuristic index based 
 policy, it is referred to as Whittle index policy.
  In \cite{Meshram15,Meshram16a}, we have considered 
a general system of a
restless multi-armed bandit with unobservable states and action
dependent transitions. In \cite{Meshram16a} we show that such a system
is \textit{approximately Whittle-indexable.} The restless bandit that
we propose in this paper is a special case of that from
\cite{Meshram16a} for which we can show exact Whittle indexability and
also obtain a closed form expression for the Whittle
index. 

The standard restless multi-armed bandit work assumes that the transition
 probabilities  and rewards are known. In recent work of \cite{LiuZhao13}, 
 UCB based learning algorithm studied for a restless multi-armed bandit   
when transition probabilities and rewards  are unknown.  
Also, Thompson sampling algorithm for a restless single armed bandit proposed 
and analysed in  \cite{Meshram16b,Meshram16c}.  
In this paper, we also propose  Thompson sampling based learning algorithm 
for RMAB and analyse its properties via numerical experiments. 

\subsection{Contributions}
This paper is an extended version of \cite{Meshram17}. Here, we extend our 
earlier work in \cite{Meshram17} to long term average reward problem and 
discuss few variants of models.  
The detailed contributions from this paper are as follows.
\begin{enumerate}
\item We develop a restless multi-armed bandit (RMAB) model for use in
  recommendation systems and playlist creation. The arms of the bandit
  correspond to the items that may be recommended. Two types of arms
  may be defined---type $A$ and type $B.$ Each `like' for a played arm
  yields a unit reward. We will seek a policy that maximises the
  infinite horizon discounted reward and a policy that maximises the long
  term average reward.  The details are in
  Section~\ref{sec:prelim}.
\item We derive the value function properties for both type of arms in 
infinite horizon discounted  reward and average reward. Using these properties, 
we obtain closed form expressions for value functions in Section~\ref{sec:whittle}.   

\item We show that both types of arms are Whittle-indexable and obtain
  closed form expressions for the Whittle index. It is a function of the
  state of the arm. This is obtained for  infinite horizon discounted reward case in Section~\ref{subsec:Whittle-discount}. 
  We also derive the expression for the Whittle
  index in average reward case, see Section~\ref{subsec:avg-whittle}.
\item The Whittle index policy is compared against a myopic policy in
  numerical experiments. We see that the index based policy indeed
  outperforms the myopic policy in many cases. This is
  covered in Section~\ref{sec:num-results}.
\item In Section~\ref{subsec:discuss} we discuss 
dual speed restless bandits for hidden states. For few  variations of type A and 
type B arms, we obtain closed form expression of Whittle index.
\item Finally, in Section~\ref{sec:learning} we provide a Thompson
  sampling based algorithm for online learning of the parameters of
  the arms. A numerical comparison of the regret shows that the
  learning is effective.

\end{enumerate}
We remark that the objective of the paper is not to design a
recommendation system but to develop a new framework with provable
properties for creating such systems.

\section{Preliminaries and Model Description}
\label{sec:prelim}
To anchor the discussion, assume an automated playlist creation system
(APCS) with $N$ items in its database. When an item is played, the
user provides a binary feedback; a possible mechanism could be by
clicking, or not clicking a `like' button. The user's interest in an
item at any time is determined by an intrinsic interest and also on
the time since it was last played. These features are captured in the
model as follows. Each item in the database corresponds to one arm of
the multi-armed bandit. The playout history of an arm is captured via
a state variable for the arm and the interest in the item is captured
via state-dependent `like'-probability for the arm. The state of each
arm evolves independently of the other arms with transition
probabilities that depend on whether it is played or not played.

We now formally describe the model. Time is measured in recommendation
steps and is indexed by $t= 1,2,3,\cdots,$ i.e., a recommendation is
made in every step. $X_t(n) \in \{0,1\}$ is the state of arm $n$ at
the beginning of step $t.$ $A_t(n) \in \{0,1\}$ is the action in step
$t$ for arm $n$ with $A_t(n) =1$ corresponding to playing arm $n$ and
$A_t(n) = 0$ corresponding to not playing it.  $X_t(n)$ evolves
according to transition probabilities that depend on $A_t(n).$ There
are two types of arms with arms in $\mathcal{A} = \{1,\ldots,M\}$
being type $A$ arms and those in $\mathcal{B} = \{M+1, \ldots,N\}$
being the type $B$ arms. $P_{ij}^{n}(a)$ denotes the transition
probability from state $i$ to state $j$ for arm $n$ under action $a.$

For type $A$ arms, i.e., for $1 \leq n \leq M,$ for $P_{00}^n(1)=1,$
$P_{10}^n(1):=1,$ $P_{01}^n(0):=p_n,$ $P_{11}^n(0):=1$ Here, $p_n$
determines the preferred `repetition rate' of arm $n.$ If $p_n$ is
small, then the user prefers a large gap between successive times that
the arm is played; if it is large then the preference is for smaller
gaps. Type $A$ arms correspond to `normal' items in that the user
prefers sufficient gap between the playing of the item.

For type $B$ arms, i.e., for $M+1 \leq n \leq N,$ the transition
probabilities are $P_{01}^n(1)=1,$ $P_{11}^n(1)=1,$ $P_{00}^n(0)=1,$
$P_{10}^n(0)=p_n.$ Type $B$ arms correspond to `viral' items where the
preference is to have it played frequently and until it is `time to
forget' it. Thus $p_n$ determines the forgetting rate for a viral
item.

When arm $n$ is played and it is in state $i,$ then a unit reward is
accrued with probability $\rho_{n,i},$ $0< \rho_{n,i} < 1.$ The reward
corresponds to the user liking the playing of the arm and $\rho_{n,i}$
represents the intrinsic preference for the item. No reward is accrued
from arms that are not played.
The transition probabilities and rewards for the two types of arms are
illustrated in Fig.~\ref{Fig:state-transitions-typical} and
Fig.~\ref{Fig:state-transitions-viral}.

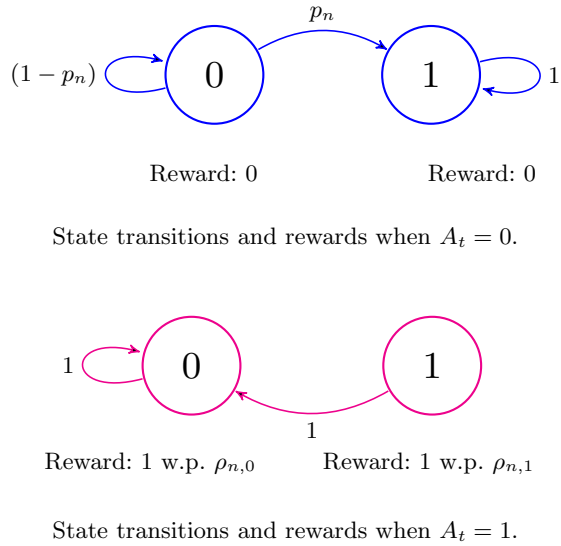
\begin{figure}
  \begin{center}
    \begin{tikzpicture}[draw=blue,>=stealth', auto, semithick, node distance=1.8cm]
      \tikzstyle{every
        state}=[fill=white,draw=blue,thick,text=black,scale=1.6]
      \node[state] (A) {$0$}; \node[state] (B)[ right of=A] {$1$};
      \path (A) edge[loop left] node{$(1-p_n)$} (A) edge[bend
        left,above,->] node{$p_n $ } (B) (B) edge[loop right]
      node{$1$} (B);
      \draw (-0.2,-1.3) node {\small{ Reward: $0$ }};
      \draw (3.5,-1.3) node {\small{ Reward: $0$ }}; 
    \end{tikzpicture}
  \end{center}

  \begin{center}
   { \small{State transitions and rewards when $A_t=0.$}}
  \end{center}
  
  \vspace{5pt}
  
  \begin{center}
    \begin{tikzpicture}[draw=magenta, >=stealth', auto, semithick, node distance=2cm]
      \tikzstyle{every state}=[fill=white,draw=magenta,thick,text=black,scale=1.6]
      \node[state]    (A)                     {$0$};
      \node[state]    (B)[ right of=A]   {$1$};
      \path
      (A) edge[loop left]     node{$1$}         (A)
      (B)     edge[bend left,below,->]      node{$1 $}         (A);    
      \draw (-0.6,-1.3) node {\small{ Reward: $1$ w.p. $\rho_{n,0}$} };
      \draw (3.1,-1.3) node {\small{ Reward: $1$  w.p. $\rho_{n,1}$} };
    \end{tikzpicture}
  \end{center}

  \begin{center}
    {\small{State transitions and rewards when $A_t=1.$}}
  \end{center}
  \caption{Transition probabilities and reward structure for type-A
    arms }
  \label{Fig:state-transitions-typical}  
\end{figure}
\begin{figure}
  \begin{center}
    \begin{tikzpicture}[draw=blue,>=stealth', auto, semithick, node distance=1.8cm]
      \tikzstyle{every state}=[fill=white,draw=blue,thick,text=black,scale=1.6]
      \node[state]    (A)                     {$0$};
      \node[state]    (B)[ right of=A]   {$1$};
      \path
      (A)  edge[loop left]     node{$1$}
      (A)
         (B)  edge[loop right]    node{$(1-p_n)$}
      (B)
       edge[bend left,below,->]      node{$p_n $}         (A); 
      \draw (-0.2,-1.3) node {\small{ Reward: $0$ }};
      \draw (3.5,-1.3) node {\small{ Reward: $0$ }}; 
    \end{tikzpicture}
  \end{center}

  \begin{center}
   { \small{State transitions and rewards when $A_t=0.$}}
  \end{center}
  
  \vspace{5pt}
  
  \begin{center}
    \begin{tikzpicture}[draw=magenta, >=stealth', auto, semithick, node distance=2cm]
      \tikzstyle{every state}=[fill=white,draw=magenta,thick,text=black,scale=1.6]
      \node[state]    (A)                     {$0$};
      \node[state]    (B)[ right of=A]   {$1$};
      \path
     (A) edge[bend left,above,->]      node{$1 $}      (B)
     (B) edge[loop right]    node{$ 1 $}   (B) ; 
      \draw (-0.6,-1.3) node {\small{ Reward: $1$ w.p. $\rho_{n,0}$} };;
      \draw (3.1,-1.3) node {\small{ Reward: $1$  w.p. $\rho_{n,1}$} };
    \end{tikzpicture}
  \end{center}

  \begin{center}
    {\small{State transitions and rewards when $A_t=1.$}}
  \end{center}
  \caption{Transition probabilities and reward structure for type-B
    arms.}
  \label{Fig:state-transitions-viral}  
\end{figure}
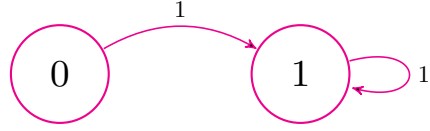

Observe that in the system, the state of an arm evolves even when it
is not played; thus this is a restless multi-armed bandit. For each
arm, we maintain a belief for the state of arm at the beginning of
time step $t,$ denoted by $\pi_t(n),$ for arm $n.$ At the end of each
time step, we can perform a Bayesian update of $\pi_t(n)$ using $p_n$
and $\rho_{n,i}$ along with the observation of the reward if the arm
is played. We will define $\pi_t(n) = \prob{X_{t}(n)
  = 0 \given H_t}.$ Here, $H_t$ denote the history of actions and observed 
  rewards up to the
beginning of time $t,$ i.e., $H_t \equiv (A_s(n), R_s(n))_{1 \leq n
  \leq N, 1 \leq s < t},$ 
and $R_t(n)$ is the reward obtained in step $n$ from arm $n.$ 
  One arm from the set of $N$ arms is to be
played at each time step.  Let $\phi=\{\phi(t)\}_{t > 0}$ be the
strategy where $\phi(t): H_t \to \{1, \ldots, N \}$ maps the history
upto time $t$ to the action of playing one of the $N$ arms at time
$t.$ Under the policy $\phi,$ let the action at time $t$ be denoted by
\begin{displaymath}
  A^{\phi}_t(n) = 
  \begin{cases}
    1 & \mbox{if }\phi(t)=n, \\
    0 & \mbox{if }\phi(t) \neq n.
  \end{cases}
\end{displaymath}
The infinite horizon expected discounted reward under policy $\phi$ is
given by
\begin{equation}
  \begin{aligned}
    && V_{\beta,\phi}(\pi) : = \expect{  \sum_{t=1}^{\infty} \beta^{t-1}
    \bigg( \sum_{n=1}^N A^{\phi}_t(n) \ (\pi_t(n) \ \rho_{n,0}
     + \  \left.(1-\pi_t(n)) \rho_{n,1} \right) \bigg)  } .
  \end{aligned}
  \label{eq:RMAB-valfn-discount}
\end{equation}
Here $\beta$ is the discount factor, $0<\beta <1,$ the initial belief
is $\pi(n) = \prob{X_1(n) = 0},$ and $\pi = [\pi(1), \cdots,
  \pi(N)]^{T}.$ 
  The long term average reward  under policy $\phi$ is given as follows.
\begin{equation}
V_{\phi}(\pi):= \lim_{T \rightarrow \infty} \frac{1}{T}
\expect{ \sum_{t=1}^{T} \bigg( \sum_{n=1}^{N} 
A^{\phi}_t(n) \ (\pi_t(n) \ \rho_{n,0}
   + \  \left.(1-\pi_t(n)) \rho_{n,1} \right)
\bigg) }
\label{eq:RMAB-valfn-avg}
\end{equation}
%
  In this paper, our goal is in a
policy $\phi$ that maximizes $V_{\beta,\phi}(\pi)$ for all $\pi \in
[0,1]^N$ assuming that we know $p_n,$ and $\rho_{n,i}$ for all $n.$
Similarly, in case of average reward problem we want to find a policy 
$\phi$ that maximizes $V_{\phi}(\pi)$ for all $\pi \in
[0,1]^N.$

\section{Towards the Whittle Index}
\label{sec:whittle}

As we have mentioned earlier, the problem
\eqref{eq:RMAB-valfn-discount} is a restless multi-armed bandit (RMAB)
with partially observable states.  In general, RMAB is computationally
intractable.  It is known to be $\mathrm{PSPACE}$-hard; see
\cite{Papadimitriou99}. In light of this hardness, heuristic policies
are sought. One class of heuristic policies is an index-based
policy. Here, at the beginning of each time step, an index is
calculated for each arm using the belief of the state of the arm, the
transition probabilities and the reward probabilities and the arm(s)
with highest index values are played in the step.  A popular
index-based policy is the Whittle-index policy based on a Lagrangian
relaxation of \eqref{eq:RMAB-valfn-discount}. This was first outlined
in \cite{Whittle88}. In many cases, this policy is known to be
asymptotically optimal; see \cite{LiuZhao10},\cite[Chapter
  $6$]{Gittins11}.  To be able to use this heuristic, we first need to
show that each arm is indexable. To effectively use it, we need to
derive the formulae to calculate the index. Indexability is proved by
analysing a single arm. 
We first analyse a single armed bandit with infinite horizon discounted 
reward problem. Later, we will examine a single armed bandit with long term
 average reward problem. 

\subsection{Discounted reward problem}
\label{subsec:Whittle-discount}
We begin by dropping the reference to $n,$ the sequence number of the
arm to simplify the notation. Next, the arm is assumed to be assigned
a \textit{subsidy} $\lambda$ for not playing it. In view of this
subsidy, \eqref{eq:RMAB-valfn-discount} may be rewritten as follows.
\begin{equation}
  \begin{aligned}
    && V_{\beta}(\pi) : = 
    \expect{ \sum_{t=1}^{\infty} \beta^{t-1}  \left( 
         A_t^{\phi} \ (\pi_t \ \rho_{0} +
            (1-\pi_t) \rho_{1})  
  + \lambda (1-A_t^{\phi}) \right)
 }.
  \end{aligned}
  \label{eq:RMAB-Single-arm-valfn-discount}
\end{equation}
Recall that $\pi_t = \prob{X_t=0 \; | \; H_t}.$ For notational
simplicity, we rewrite $A_t^{\phi}$ as $A_t$ with policy $\phi.$ The
Bayesian updates for $\pi_t$ are obtained as follows.
\begin{itemize}
\item For a type $A$ arm: If $A_t=1,$ then $\pi_{t+1} = 1$ and if
  $A_t=0,$ then $\pi_{t+1} = (1-p)\pi_t.$
\item For a type $B$ arm: If $A_t=1,$ then $\pi_{t+1} = 0$ and if
  $A_t=0,$ then $\pi_{t+1} = \pi_t + p(1-\pi_t).$
\end{itemize}
If $A_t =1,$ then the expected reward in the step is $\pi_t \rho_0 +
(1-\pi_t) \rho_1.$ The policy $\phi(t): H_t \rightarrow \{0,1\},$ maps
the history up to time $t,$ to an action $A_t$ in $t.$ From
\cite{Ross71,BertsekasV195,BertsekasV295}, the following is well
known.
\begin{itemize}
\item $\pi_t$ captures the information in $H_t,$ and is a sufficient
  statistic to construct policies that depend on the history.
\item Optimal strategies can be restricted to stationary Markov
  policies.
\item The optimum value function for fixed $\lambda$ and $\beta$,
  denoted $V_{\beta}(\pi,\lambda),$ is determined by solving the
  following dynamic program.

  { \footnotesize
  \begin{eqnarray}
    V_{1,\beta}(\pi,\lambda) &:=& 
    \begin{cases}
      \pi \rho_0 + (1-\pi) \rho_1 + \beta V_{\beta}(1,\lambda) 
      & \mbox{for type A} \\
      \pi \rho_0 + (1-\pi) \rho_1 + \beta V_{\beta}(0,\lambda) 
      & \mbox{for type B}       
    \end{cases}
    \nonumber \\ 
    V_{0,\beta}(\pi,\lambda) &:=& 
    \begin{cases}
      \lambda + \beta V_{\beta}((1-p)\pi,\lambda) 
      & \mbox{for type A} \\
      \lambda + \beta V_{\beta}(\pi+p(1-\pi),\lambda) 
      & \mbox{for type B}       
    \end{cases}
    \nonumber \\
    V_{\beta}(\pi,\lambda) &=& \max \{V_{1,\beta}(\pi,\lambda),
    V_{0,\beta}(\pi,\lambda)\}.
    \label{eq:dynamic_prog}
  \end{eqnarray}
  }
  $V_{i,\beta}(\pi,\lambda)$ is the optimal value function if $A_1=i,$ $i=0,1.$
\end{itemize}
We next derive properties of the value functions $V_{\beta},$  $V_{0,\beta},$
and $V_{1,\beta}$ in the following Lemma.
\begin{lemma}
\item 
  \begin{enumerate}
  \item For fixed $\lambda$ and $ \beta,$ $V_{\beta}(\pi,\lambda),$
    $V_{0, \beta}(\pi,\lambda),$ $V_{1,\beta}(\pi,\lambda)$ are
    non-increasing convex in $\pi.$ Furthermore,
    $V_{1,\beta}(\pi,\lambda)$ is linear in $\pi.$
  \item For fixed $\pi$ and $\beta,$ $V_{\beta}(\pi,\lambda),$
    $V_{1,\beta}(\pi,\lambda)$ and $V_{0,\beta}(\pi,\lambda)$ are
    non-decreasing convex in $\lambda.$
  \item For fixed $\pi$ and $\beta,$ $V_{1,\beta}(\pi,\lambda)$ and
    $V_{0,\beta}(\pi,\lambda)$ intersect at least once. This leads us
    to define~(1)~$\lambda_{L,\beta}$ such that for all $\lambda <
    \lambda_{L,\beta},$ the optimal action is to play the arm for all
    $\pi \in[0,1],$~and~(2)~$\lambda_{H,\beta} = \lambda_H$ such that
    for all $\lambda > \lambda_{H}$ the optimal action is to not play
    the arm for all $\pi \in [0,1].$ $\lambda_{L,\beta}$ and
    $\lambda_H$ are given as follows.
    \begin{eqnarray*}
      \lambda_H &=& 
      \begin{cases}
        \rho_1 & \mbox{for type A,} \\
        \rho_1 & \mbox{for type B,}
      \end{cases}
      \nonumber \\
      \lambda_{L,\beta} &=& 
      \begin{cases}
        \rho_0 + \beta q (\rho_0 -\rho_1) & \mbox{for type A,} \\
        \rho_1 + (1-\beta) (\rho_0- \rho_1) & \mbox{for type B.}
        \label{eq:range-lambda}
      \end{cases}
    \end{eqnarray*}

  \end{enumerate}
  \label{lemma:struct-results-1}
\end{lemma}
The proof of Lemma~\ref{lemma:struct-results-1} is analogous to the proofs in
 \cite[Lemma $2$ and $3$]{Meshram16a}.
 Hence we omit the proof.
  In the next Lemma, we  state the Lipschitz properties of value  function with respect to $\pi$ and $\lambda.$ 
\begin{lemma}
\item 
  \begin{enumerate}
  \item For fixed $\lambda$ and $\beta,$ and $\forall \pi_1, \pi_2
    \in [0,1],$ we have 
    \begin{eqnarray*}
      \big\vert V_{\beta}(\pi_1,\lambda) -
      V_{\beta}(\pi_2,\lambda) \big\vert \leq
      (\rho_1-\rho_0) |\pi_1 -\pi_2|
    \end{eqnarray*}
  \item For fixed $\pi \in[0,1],$ and for $0<\beta<1,$ \\
    $$
    \frac{\partial V_{\beta}(\pi,\lambda)}{\partial \lambda}, \ \ 
    \frac{\partial V_{1,\beta}(\pi,\lambda)}{\partial \lambda}, \ \ 
    \mbox{ and } \frac{\partial V_{0,\beta}(\pi,\lambda)}{\partial
      \lambda}
    $$
    are bounded above by $\frac{1}{1-\beta}.$
\end{enumerate}  
  \label{lemma:struct-results-2}
\end{lemma}
The proof is given in Appendix~\ref{proof:struct-results-2}.
\begin{remark}
It is possible that $V_{\beta}(\pi, \lambda)$ is not differentiable with respect to 
$\pi$ or $ \lambda.$ In that case the partial derivative of $V_{\beta}(\pi, \lambda)$
 should be taken to be the right partial derivative. 
 Note that such the right partial derivative exists because $V_{\beta}(\pi, \lambda)$ 
 is  convex  in  $\pi,$ $\lambda,$ and bounded. 
\end{remark}

Using the Lipschitz property of $V_{\beta}(\pi,\lambda)$ in $\pi,$ we derive the next result.
\begin{lemma}
 For fixed $\lambda$ and $\beta,$ $V_{1,\beta}(\pi,\lambda) -
    V_{0,\beta}(\pi,\lambda)$ is decreasing in $\pi.$
 \label{lemma:v-pi-monotone}
\end{lemma} 
The proof is detailed  in Appendix~\ref{proof:v-pi-monotone}.
We now ready to present our first main result on a threshold policy structure.
%
%
\begin{theorem}[Single threshold policy]
  For the single-armed bandit, $0<\beta <1,$ and $\lambda_L \leq
  \lambda \leq \lambda_{H,\beta},$ the optimal policy is of threshold
  type with a single threshold.  That is, there is a unique threshold
  $\pi_T \in[0,1]$ such that
  \begin{eqnarray*}
    V_{\beta}(\pi,\lambda) = 
    \begin{cases}
      V_{1,\beta}(\pi,\lambda) & \mbox{if $\pi \leq \pi_T,$ } \\
      V_{0,\beta}(\pi,\lambda) & \mbox{if $\pi \geq \pi_T,$ }
    \end{cases}
  \end{eqnarray*}
  where $\pi_T = \{ \pi \in [0,1]: V_{1,\beta}(\pi,\lambda) =
  V_{0,\beta}(\pi,\lambda)\}.$
  \label{thm:single-threshold-type-A}
\end{theorem}
\begin{proof}
  Fix $\beta, \lambda.$ From Lemma \ref{lemma:struct-results-1} and
  \ref{lemma:v-pi-monotone}, we have the following. For a fixed
  $\lambda,$ $V_{1,\beta}(\pi,\lambda)$ is linear in $\pi$ and
  $V_{0,\beta}(\pi,\lambda)$ is convex in $\pi.$ Also,
  $V_{1,\beta}(\pi,\lambda) -V_{0,\beta}(\pi,\lambda)$ is decreasing
  in $\pi.$ Thus there is at most one threshold.  Furthermore,
  $V_{1,\beta}(\pi,\lambda)$ and $V_{0,\beta}(\pi,\lambda)$ intersect
  at least once for $\lambda_{L,\beta}\leq\lambda \leq \lambda_H,$
  $\pi \in [0,1].$ This completes the proof.
\qed
\end{proof}
\begin{remark}
\item 
\begin{enumerate}
\item 
The threshold policy implies that whenever belief at time step $t,$ $\pi_t$
 is greater than $\pi_T,$ then the optimal action is to not play the item.
\item For type-A arm, the belief about the state $0$  evolves to $(1-p) \pi_t$
 when arm is not played.  Thus, $\pi_{t+1}$ decreases  whenever the 
 item is not played and after some time steps, the optimal action will be to play 
 the item. Once item is played, then $\pi_{t+1}$  reaches the state $0$ with probability $1.$ Let $K$ denotes the number  of time steps  to wait to play that item again and it  depends on  $\pi_T,$  $\pi$ and $p.$ It is defined as follows.
\[
K(\pi,\pi_T) := \min \{k \geq 0: (1-p)^k\pi < \pi_T \}.
\]
Then, 
\begin{eqnarray}
K(\pi, \pi_T) =
\begin{cases}
0 \ \ \ \ \  \  \ \ \ \ \  \ \ \ \ \ \ \ \ \ \ \ \ \mbox{if $\pi<\pi_T,$ } \\
\left \lfloor {\frac{\log(\frac{\pi_T}{\pi})}{\log(1-p)}}  \right \rfloor+1 \ \ \ \ \  \mbox{ if $\pi \geq \pi_T$.}
\end{cases}
\label{eq:K-step-wait-normal}
\end{eqnarray}
Using a threshold policy result and $K(\pi,\pi_T),$ we can derive  the value function  expressions as follows. 
\begin{eqnarray}
V_{1,\beta}(\pi,\lambda) &=& \pi \rho_0 + (1-\pi) \rho_1 + \beta V_{0,\beta} (1, \lambda), \nonumber \\  \nonumber \\
V_{0,\beta}(\pi,\lambda) &=& \frac{\lambda \left( 1- \beta^{K(\pi,\pi_T)}\right)}{1-\beta} +  \beta^{K(\pi,\pi_T)+1} V_{0,\beta}(1,\lambda) 
 + \beta^{K(\pi,\pi_T)} \times \nonumber \\  
&& \ \ \ \ \  \hspace{50pt}\left(\rho_1 + (\rho_0 - \rho_1) (1-p)^{K(\pi,\pi_T)} \pi \right),  \nonumber \\ \nonumber \\
V_{0,\beta}(1,\lambda) &=& \frac{\lambda (1-\beta^{K(1, \pi_T)}) }{\left(1-\beta^{K(1,\pi_T)+1}\right)(1-\beta)}
+ \beta^{K(1,\pi_T)}  
\frac{ \left(\rho_1 + (\rho_0 - \rho_1)(1-p)^{K(1,\pi_T)} \right)}{\left(1-\beta^{K(1,\pi_T)+1}\right)}. \nonumber \\ 
\label{eq:value-fun-type-a}
\end{eqnarray}
\item For type-B arm, the belief about the state $0$ evolves to 
$\pi_t + p (1-\pi_t)$  when arm is not played. 
Observe that $\pi_{t+1}$ is increases whenever arm is not played. 
 Define $T(\pi) := \pi + p (1-\pi),$ $T^k(\pi) = T^{k-1}(T(\pi))$ and 
 $\lim_{k \rightarrow \infty} T^k(\pi) =1.$ 
When a threshold $\pi_T \in (0,1],$ note that $V_{\beta}(0,\lambda) = 
V_{1,\beta}(0,\lambda)= \frac{\rho_1}{1-\beta}.$  
Also, following  holds.
\begin{eqnarray*}
V_{\beta}(T(\pi),\lambda)= 
\begin{cases}
 V_{1,\beta}(T(\pi),\lambda) & \mbox{if $T(\pi) < \pi_T,$} \nonumber \\
 V_{0,\beta}(T(\pi),\lambda) & \mbox{if $T(\pi) \geq \pi_T.$ }
\end{cases}
\end{eqnarray*}
This discussion suggests that once that item is played, the user keeps liking that item. If the 
item is not played to the user, i.e. $\pi_t \geq \pi_T$, then  the state of the item $\pi_t$ is always greater than $\pi_T$ for all $t$  and that means, that item is not played to the user at all. 
 Thus the value function expressions are described below. 
\begin{eqnarray}
V_{1,\beta}(\pi,\lambda) &=& \pi \rho_0 + (1-\pi) \rho_1 + \beta \frac{\rho_1}{1-\beta}, \nonumber  \\
V_{0,\beta}(\pi, \lambda) &=& 
\begin{cases}
\lambda + \beta (\rho_0 T(\pi) + (1-T(\pi)) \rho_1) + \beta \frac{\rho_1}{1-\beta} 
& \mbox{ if  $T(\pi) < \pi_T,$} \\
\frac{\lambda}{1-\beta} & \  \mbox{if $T(\pi) \geq \pi_T,$}  
\end{cases}
\label{eq:value-fun-type-b-piTnonz}
\end{eqnarray}
for $\pi_T \in (0,1].$
When $\pi_T = 0,$ we have 
\begin{eqnarray*}
V_{1,\beta}(\pi,\lambda) &=& \pi \rho_0 + (1-\pi) \rho_1 + \beta 
\frac{\lambda}{1-\beta}, \\
V_{0,\beta}(\pi, \lambda) &=& \frac{\lambda}{1-\beta}.
\end{eqnarray*}
\end{enumerate}
\label{rem:threshold-policy}
\end{remark}
From Theorem~\ref{thm:single-threshold-type-A} and Remark~\ref{rem:threshold-policy},
 we can show the following.
\begin{lemma}
  For fixed $\pi$ and $\beta,$ $V_{1,\beta}(\pi,\lambda) -
  V_{0,\beta}(\pi,\lambda)$ is a decreasing in $\lambda.$
  \label{lemma:struct-result-3}
\end{lemma}
The proof is given in Appendix~\ref{proof:struct-results-3}.
We  first define indexability  and later show that type-A and type-B arms are indexable.
 Define,
\begin{eqnarray*}
  \mathcal{P}(\lambda) &:= & \left\{ \pi \in[0,1]:
  V_{1,\beta}(\pi,\lambda) \leq V_{0,\beta}(\pi,\lambda) \right\}
\end{eqnarray*}
$\mathcal{P}(\lambda)$ is a set of all $\pi$ for which the optimal
action is to not play the arm.
\begin{definition}[Whittle indexability, \cite{Whittle88}]
  An arm is Whittle indexable if $\mathcal{P}(\lambda)$ monotonically
  increases from $\emptyset$ to the entire state space $[0,1]$ as
  $\lambda$ increases from $-\infty$ to $\infty$, i.e.,
  $\mathcal{P}(\lambda_1) \setminus \mathcal{P}(\lambda_2) =
  \emptyset$ whenever $\lambda_1 < \lambda_2$. Further, a multi-armed
  bandit with $N$ arms is indexable if all arms are indexable.
  \label{def:indexable}
\end{definition}
We require the following result from \cite[Lemma $4$]{Meshram16a} to
prove Whittle indexability.
\begin{lemma}
  Let $\pi_T(\lambda) = \inf\{0 \leq \pi \leq 1: V_{S,\beta}(\pi,
  \lambda) = V_{NS,\beta}(\pi, \lambda) \} \in [0,1]$. If 
  \begin{equation}
    \frac{\partial V_{1,\beta}(\pi,\lambda)}{\partial \lambda} \bigg 
    \rvert_{\pi = \pi_T(\lambda)} \ < \
    \frac{\partial V_{0,\beta}(\pi,\lambda)}{\partial \lambda} \bigg 
    \rvert_{\pi=\pi_T(\lambda)},
    \label{eq:grad-VS-less-grad-VNS1}
  \end{equation}
  then $\pi_T(\lambda)$ is a monotonically decreasing function of
  $\lambda.$
  \label{lemma:indexability}
\end{lemma}
We next present our second main result.

\begin{theorem}[Whittle indexable]
  The single-armed bandit is indexable for $0 < \beta< 1$ and
  $\lambda_{L,\beta} \leq \lambda \leq \lambda_H.$
  \label{thm:indexability}
\end{theorem}

\begin{proof}
  From Definition \ref{def:indexable}, it is clear that we have to
  show $\mathcal{P}(\lambda_1) \subseteq \mathcal{P}(\lambda_2)$
  whenever $\lambda_2 > \lambda_1.$ From Lemma
  \ref{lemma:struct-result-3}, we note that $V_{1,\beta}(\pi,\lambda)
  -V_{0,\beta}(\pi,\lambda)$ is decreasing in $\lambda$ for fixed
  $\pi$ and $\beta.$ Therefore, \eqref{eq:grad-VS-less-grad-VNS1}
  holds true. Using Lemma \ref{lemma:indexability}, $\lambda_2 >
  \lambda_1$ implies $\pi_T(\lambda_2) < \pi_T(\lambda_1)$ for fixed
  $\beta.$ Hence, from the definition of the set
  $\mathcal{P}(\lambda),$ we get $\mathcal{P}(\lambda_1) \subseteq
  \mathcal{P}(\lambda_2)$ whenever $\lambda_2 > \lambda_1.$ This
  completes the proof.
  \qed
\end{proof}
We are now ready to define the Whittle index for an arm and provide an
explicit formula for Whittle index in case of both type $A$ and type $B$ arms.
\begin{definition}[Whittle index]
  If an arm is indexable and is in state $\pi,$ then its Whittle
  index, $W(\pi),$ is
  \begin{eqnarray*}
    W(\pi) &:=& \inf_{\lambda}\{\lambda: V_{1,\beta}(\pi,\lambda) =
    V_{0,\beta}(\pi,\lambda) \}.
  \end{eqnarray*}
  \label{def:whittleind}
\end{definition}
$W(\pi)$ is the minimum subsidy $\lambda$ such that the optimal action
is to not play the are at the given $\pi.$ To compute the Whittle
index, we have to obtain the expressions of $V_{1,\beta}(\pi,\lambda)$
and $V_{0,\beta}(\pi,\lambda),$ equate them and solve it for $\lambda.$
After simplification, the Whittle index for type-A arm  is as follows.
  {\small{
  \begin{eqnarray}
    W(\pi) = 
    \rho_1 + \beta^{(K+1)} (\rho_0-\rho_1)(1-p)^{K} + 
     \frac{(\rho_0-\rho_1)}{(1-\beta)} 
    \left( 1 -\beta^{(K+1)} \right) \times 
    \left[ (1-\beta(1-p)) \pi\right]. \nonumber \\
    \label{eq:whittle-index-A}
  \end{eqnarray} 
  }}
  Here, $K= K(1,\pi)$ is waiting time before playing  that arm again.  
   Similarly, we can obtain the Whittle index formula for type-B arm, and it is
\begin{eqnarray}
W(\pi) &=&
\begin{cases}
 \rho_1 + (1-\beta) (\rho_0 - \rho_1) \pi & \mbox{if $\pi \in (0,1],$} \nonumber \\
 \rho_1 + (\rho_0 - \rho_1) \pi & \mbox{if $\pi = 0.$} \nonumber 
 \end{cases} 
 \\
 \label{eq:whittle-index-B}
\end{eqnarray}
\begin{remark}
  Note that the Whittle index of an arm depends on model parameters
  $p,\rho_0,\rho_1,$ discount parameter $\beta,$ and belief $\pi.$
\end{remark}
%
\subsection{Average reward problem}
\label{subsec:avg-whittle}
We rewrite the average reward problem~\eqref{eq:RMAB-valfn-avg}
in the view of subsidy $\lambda$ as follows.
\begin{eqnarray}
V_{T, \phi} (\pi) &:=& \expect{\sum_{t=1}^{T} A^{\phi}_t  \left( \pi_t \rho_0
+ (1-\pi_t) \rho_1 \right) + \lambda (1- A_t^{\phi})}, \nonumber \\ 
V_{\phi}(\pi) &=& \max_{\phi} \lim_{T \rightarrow \infty} \frac{1}{T}V_{T, \phi}(\pi). 
\label{eq:avg-reward-single-subsidy}
\end{eqnarray}
Here, $A^{\phi}(t) = i$ if $\phi(t) =i,$ $i \in \{0,1\},$ $\pi_1 = \pi.$
It is solved by the vanishing discount approach \cite{Avrachenkov15,Ross93}---by
first considering a discounted reward system and then taking limits as
the discount approaches to $1.$  
Define 
\begin{eqnarray*}
\overline{V}_{\beta}(\pi, \lambda) :=
\begin{cases}
 V_{\beta}(\pi, \lambda) - V_{\beta}(1,\lambda) & \mbox{for type-A arm,} \\
 V_{\beta}(\pi, \lambda) - V_{\beta}(0,\lambda) & \mbox{for type-B arm,}
\end{cases}
\end{eqnarray*}
for $\pi \in [0,1].$ Using Eqn.~\eqref{eq:dynamic_prog},  we can obtain for type-A arm 
\begin{eqnarray}
  \overline{V}_{\beta}(\pi, \lambda) + (1- \beta)V_{\beta}(1,\lambda)   
  &=&  \max \left\{ 
    \lambda + \beta \overline{V}_{\beta}((1-p)\pi, \lambda),   \pi \rho_0 + (1-\pi) \rho_1 \right\}, 
    \nonumber  \\
\label{eqn:V-bar-of-beta-1}
\end{eqnarray}
and for type-B arm
\begin{eqnarray}
  \overline{V}_{\beta}(\pi, \lambda) + (1- \beta)V_{\beta}(0,\lambda)  
   &=&  \max \left\{ 
    \lambda + \beta \overline{V}_{\beta}(\pi + p (1-\pi),\lambda),   \pi \rho_0 + (1-\pi) \rho_1 \right\}.
    \nonumber \\
\label{eqn:V-bar-of-beta-2}
\end{eqnarray}
From Lemma~\ref{lemma:struct-results-1},  $\overline{V}_{\beta}(\pi, \lambda)$
is a convex monotone function in $\pi$  for fixed  $\lambda$ and $\beta.$ 
By definition of 
$\overline{V}_{\beta} (\pi, \lambda)$ we have $\overline{V}_{\beta}(1, \lambda) = 0$ for type-A and $\overline{V}_{\beta}(0, \lambda) = 0$ for type-B arm.  Further, from  
Lemma~\ref{lemma:struct-results-2},  we know that
there is a constant $C < \infty$ such that
$\big\vert V_{\beta}(\pi, \lambda) - V_{\beta}(1,\lambda) \big\vert < C$
for fixed $\beta$ and  $\lambda \in [-\rho_1, \rho_1].$
 This implies
that $\overline{V}_{\beta}(\pi, \lambda)$ is bounded and Lipschitz-continuous.
Finally, $(1-\beta)V_{\beta}(\pi,\lambda)$ is also bounded.  Hence we can
apply the Arzela-Ascoli theorem \cite{rudin-principles}, to find a
subsequence $(\overline{V}_{\beta_k}(\pi,\lambda), (1-\beta_k)V_{\beta_k}(\pi,\lambda))$
that converges uniformly to $(V(\pi,\lambda),g)$ as $\beta_k \rightarrow 1.$
Thus, as $\beta_k \rightarrow 1,$ along an appropriate subsequence,
\eqref{eqn:V-bar-of-beta-1}   reduces to
\begin{eqnarray}
  V(\pi,\lambda) + g = \max \left\{\lambda+ V((1-p)\pi,\lambda), \pi \rho_0 + (1-\pi) \rho_1 \right\}, 
  \label{eq:dynamic-prog-avgc-A}
\end{eqnarray}
and~\eqref{eqn:V-bar-of-beta-2} reduces to
\begin{eqnarray}
  V(\pi,\lambda) + g = \max \left\{\lambda+ V(\pi + p (1-\pi),\lambda), \pi \rho_0 + (1-\pi) \rho_1 \right\}.
  \label{eq:dynamic-prog-avgc-B}
\end{eqnarray}
 Equation~\eqref{eq:dynamic-prog-avgc-A} and \eqref{eq:dynamic-prog-avgc-B} are the dynamic programming equations for type-A and type-B arm in case of average reward system.
Hence it is the optimal solution of \eqref{eq:avg-reward-single-subsidy}. \\
Also, note that $V(\pi,\lambda)$ inherits the structural properties of 
$V_{\beta}(\pi, \lambda).$ From Lemma~\ref{lemma:struct-results-1} and Theorem~\ref{thm:single-threshold-type-A}, we obtain next result. 
\begin{lemma}
\item 
\begin{enumerate}
\item For fixed $\lambda,$ $V(\pi, \lambda)$ is a monotone non-increasing and convex in $\pi.$ 
\item The optimal policy is  a single threshold type
  for $\lambda \in [\lambda_L, \lambda_H],$ where $\lambda_H = \rho_1$ 
  and 
  \begin{eqnarray*}
  \lambda_L &=&
  \begin{cases}
  \rho_0 +  p (\rho_ 0 - \rho_1) & \mbox{for type A,}\\
  \rho_1 &\mbox{for type B.}
  \end{cases}
\end{eqnarray*}   
\end{enumerate}
\label{lemma:avg-threshld-special}
\end{lemma}
This in turn leads us to the following theorem which is a direct analog of
Theorem~6.17 in \cite{Ross93}.
\begin{theorem}
  If there exists a bounded function $V(\pi,\lambda)$ for $\pi \in [0,1],$ 
  $\lambda \in [\lambda_L, \lambda_H]$ and
  a constant $g$ that satisfies equation~\eqref{eq:dynamic-prog-avgc-A},
   \eqref{eq:dynamic-prog-avgc-B} then
  there exists a stationary policy $\phi^*$ such that
  \begin{eqnarray}
    g &=& \max_{\phi}\lim_{T \rightarrow \infty} \frac{1}{T} V_{T,\phi}(\pi,\lambda) \\
    \phi^{*} &=& \arg \max_{\phi}\lim_{T \rightarrow \infty} \frac{1}{T} V_{T,\phi}(\pi,\lambda).
  \end{eqnarray}
  for all $\pi \in [0,1],$ fixed $\lambda,$ and moreover, $\phi^*$ is the policy for
  which the RHS of \eqref{eq:dynamic-prog-avgc-A},\eqref{eq:dynamic-prog-avgc-B} is maximized.
  \label{thm:avg-reward-opt-policy}
\end{theorem}
We next  derive the Whittle index formula. 
From Lemma~\ref{lemma:indexability}, recall that  to  claim indexability, we have shown that $\pi_T(\lambda)$   is a monotonically decreasing in $\lambda$ for discounted reward case. Similarly, in average reward case, we require to show this claim. 
From Arzela-Ascoli theorem, we know that $V(\pi,\lambda)$ inherits the properties of $V_{\beta}(\pi, \lambda).$  
Thus, $\pi_T(\lambda)$ is a  monotonically decreasing in $\lambda$ for $\beta =1.$  Then we can show that an arm is indexable. Further, the index can be evaluated  by letting 
$\beta \rightarrow 1$ in the Whittle index formula of discounted case. 
Hence from equation~\eqref{eq:whittle-index-A} and  \eqref{eq:whittle-index-B}, we obtain  
the Whittle index formula.
\begin{itemize}
\item For type-A arm:
\begin{eqnarray}
W(\pi) = 
\rho_1 +  (\rho_0-\rho_1)(1-p)^{K} + 
K (\rho_0-\rho_1) 
 \left[ (1-(1-p)) \pi\right]. \nonumber \\ 
\label{eq:whittle-index-A-avg}
\end{eqnarray} 
Here, $K = K(1,\pi)$ is waiting time before playing  that arm again.  
\item For type-B arm: 
\begin{eqnarray}
W(\pi) &=& \rho_1.
 \label{eq:whittle-index-B-avg}
\end{eqnarray}
\end{itemize}
%
\section{Numerical Results: Whittle index and Myopic Algorithm}
\label{sec:num-results}

\begin{algorithm}[bt]
  \caption{Whittle index algorithm for APCS with type $A$ and type $B$
    arms.}
  \begin{algorithmic}
    \STATE Input: $N$ arms, initial belief $\pi =
    \left[\pi(1), \cdots, \pi(N)\right],$ $\pi_1 = \pi,$
    \REPEAT
    \FOR{$n=1,\ldots,N$}
    \STATE Compute $W(\pi(n))$ 
    \ENDFOR
    \STATE \textbf{Evaluate}  $i = \arg\max_{ 1 \leq n \leq N} W(\pi(n)).$
    \STATE \textbf{Play} arm $i$ 
    \IF{ $i$ is Type A}
    \STATE $\pi_{t+1}(i) = 1$      
    \ENDIF
    \IF{$I$ is Type B}
    \STATE $\pi_{t+1}(i) = 0$
    \ENDIF
    \FOR{$n=1, \ldots, M,$ $n \neq i$}
    \STATE $\pi_{t+1}(n) = (1-p_n) \pi_t(n),$
    \ENDFOR
    \FOR{$n=M+1, \ldots, N,$ $n \neq i$}
    \STATE $ \pi_{t+1}(n) = (1-p_n) \pi_t(n) + p_n ,$ 
    \ENDFOR
    \UNTIL{forever}
 \end{algorithmic}
  \label{algo:Whittle-index-typical-viral}
\end{algorithm}

\begin{figure*}
  \centering
  \subfloat[]{ \includegraphics[width=0.6\columnwidth]{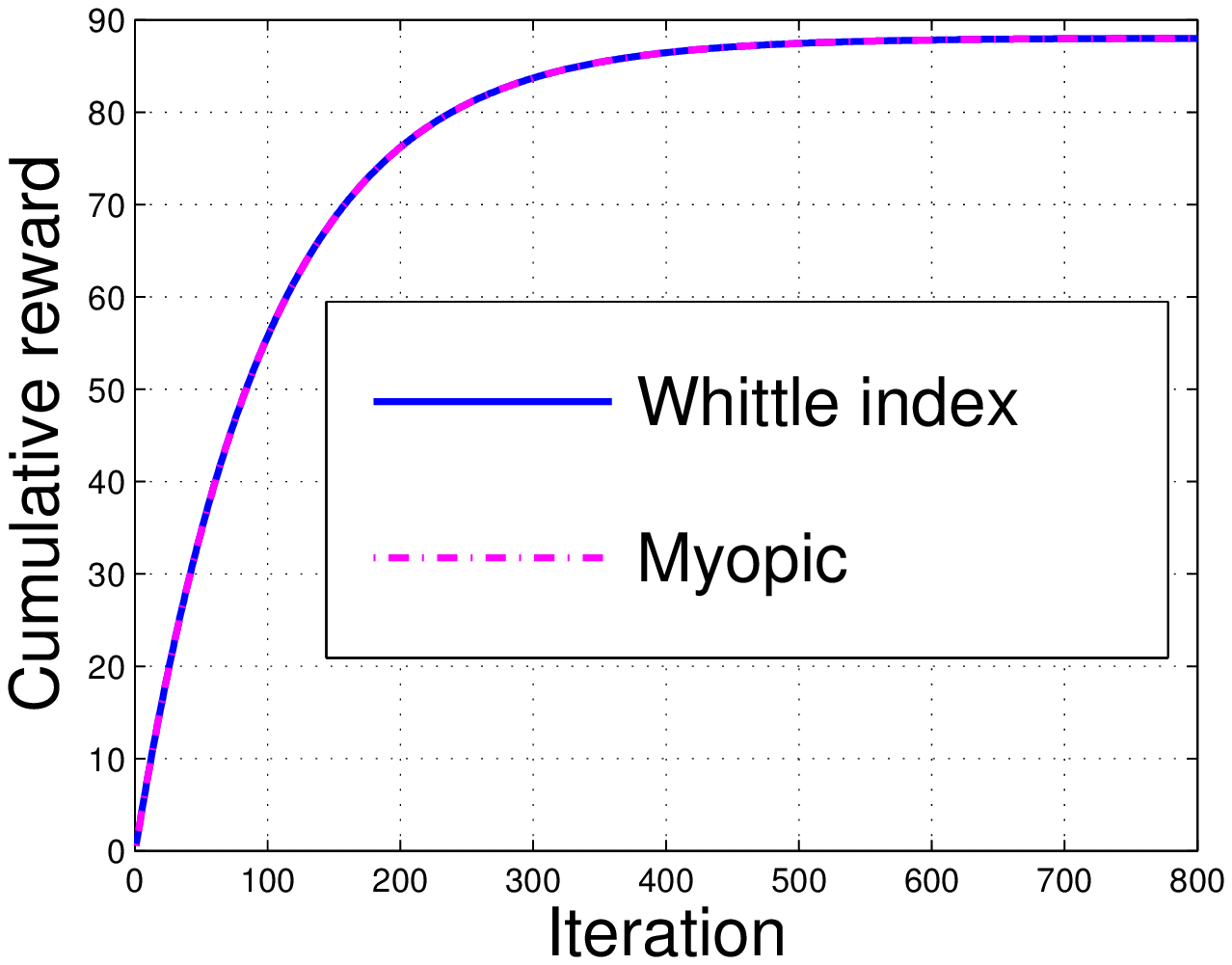}
  \label{Fig-CR-1-a}
  }
 \\
\subfloat[]{ \includegraphics[width=0.6\columnwidth]{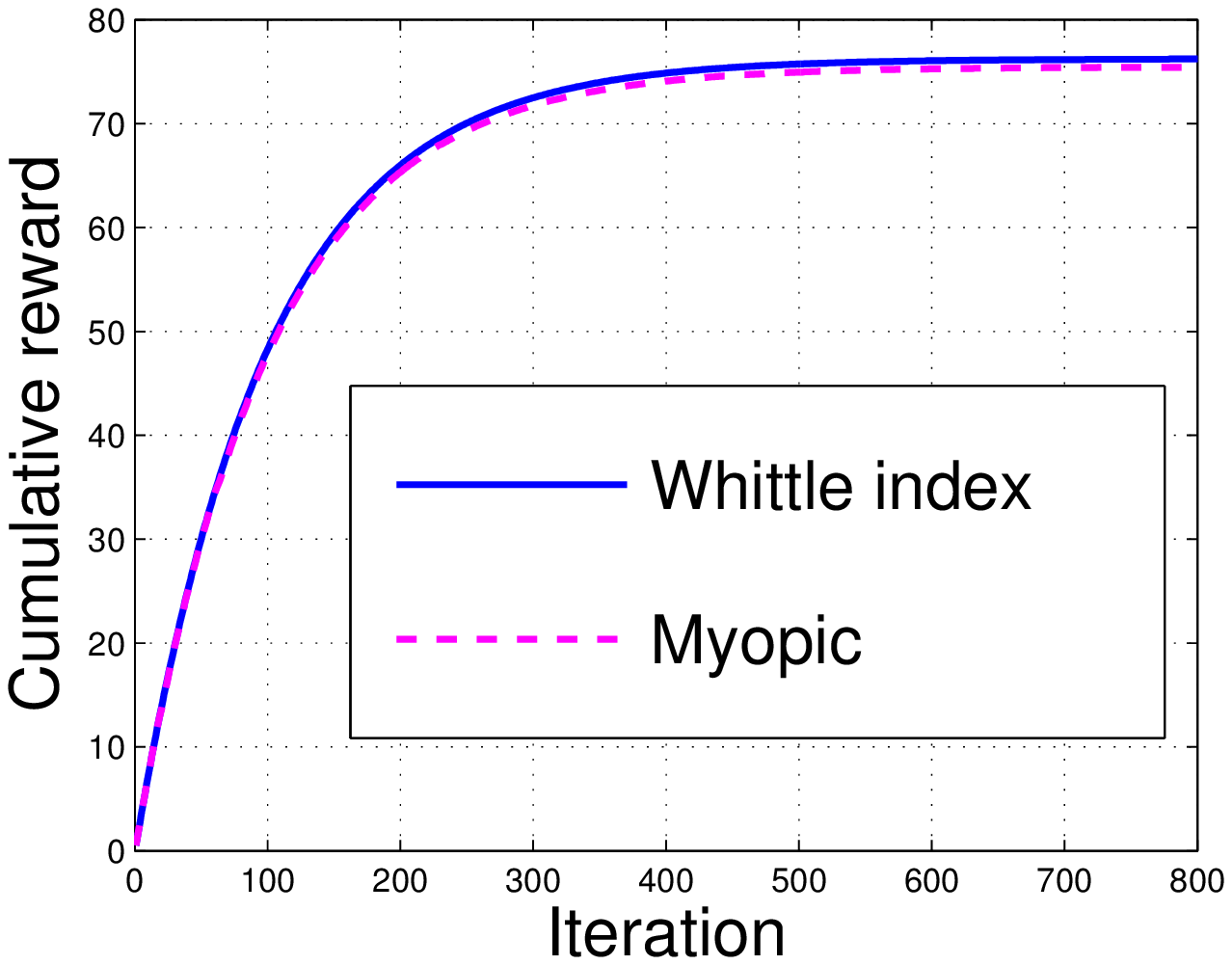}
  \label{Fig-CR-1-b} } 
  \\
\subfloat[]{ \includegraphics[width=0.6\columnwidth]{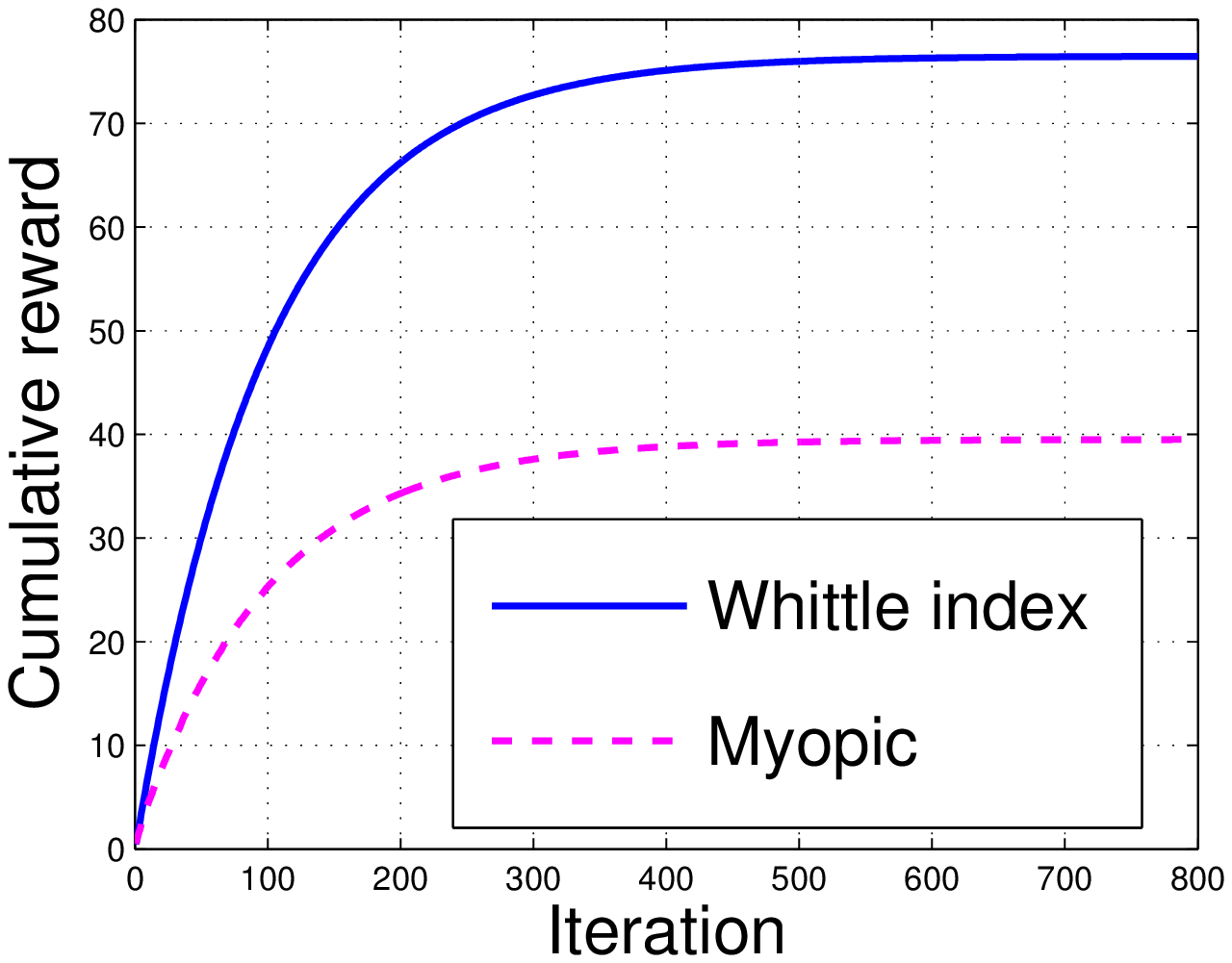}
  \label{Fig-CR-1-c} } 
\vspace{0.1in}

    Parameters

    \vspace{0.1in}
    
    \begin{tabular}{|c|c|c|c|}
      \hline
      & $\rho_0$  & $\rho_1$ & $p$ \\
      \hline 
      (a) & $[0.07, 0.04, 0.05, 0.12, 0.99]$ &
      $[0.71, 0.85, 0.77, 0.76, 0.88]$ &
      $[0.09, 0.23, 0.23, 0.12, 0.27]$\\
      \hline 
      (b) & $[0.02, 0.02, 0.11, 0.16, 0.19]$ &
      $[0.64, 0.77, 0.74, 0.60,0.76]$ &
      $[0.06, 0.24, 0.10, 0.16,0.15]$\\
      \hline
      (c) & $[0.07, 0.09, 0.01, 0.19,0.04]$ &
      $[0.63, 0.71, 0.66, 0.75,0.77]$ &
      $[0.29, 0.28, 0.03, 0.22,0.18]$\\
      \hline
    \end{tabular}

    \caption{Cumulative reward vs time for $N = 5$ with APCS having
      both types of arms.}
  \label{Fig-CR-1}
\end{figure*}

\begin{figure*}
\centering
  \subfloat[]{ \includegraphics[width=0.6\columnwidth]{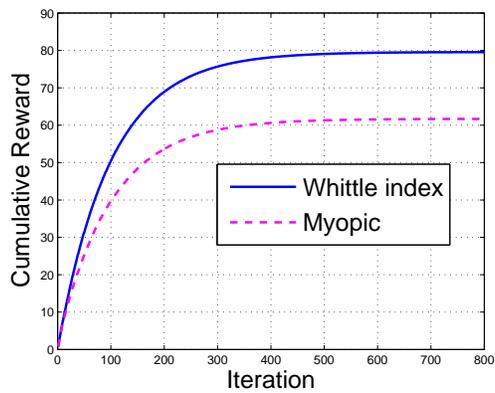}
  \label{Fig-CR-2-a}
  }
 \\
\subfloat[]{ \includegraphics[width=0.6\columnwidth]{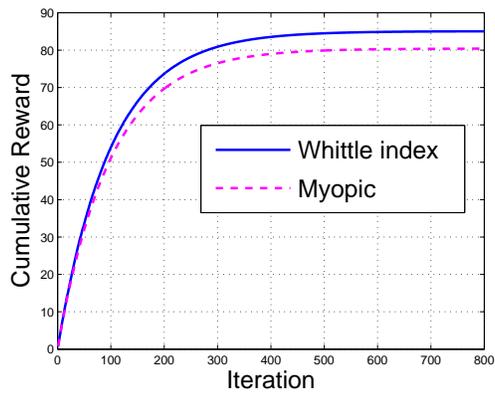}
  \label{Fig-CR-2-b} } 
  \\
\subfloat[]{ \includegraphics[width=0.6\columnwidth]{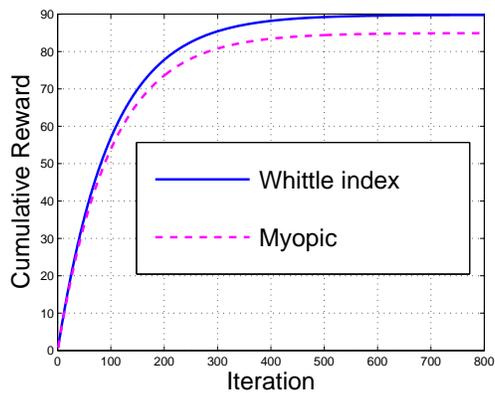}
  \label{Fig-CR-2-c} } 
  \caption{Cumulative reward vs time for different $N$ with APCS
    having both types of arms.}
  \label{Fig-CR-2}
\end{figure*}

\begin{figure*}
\centering
  \subfloat[]{ \includegraphics[width=0.6\columnwidth]{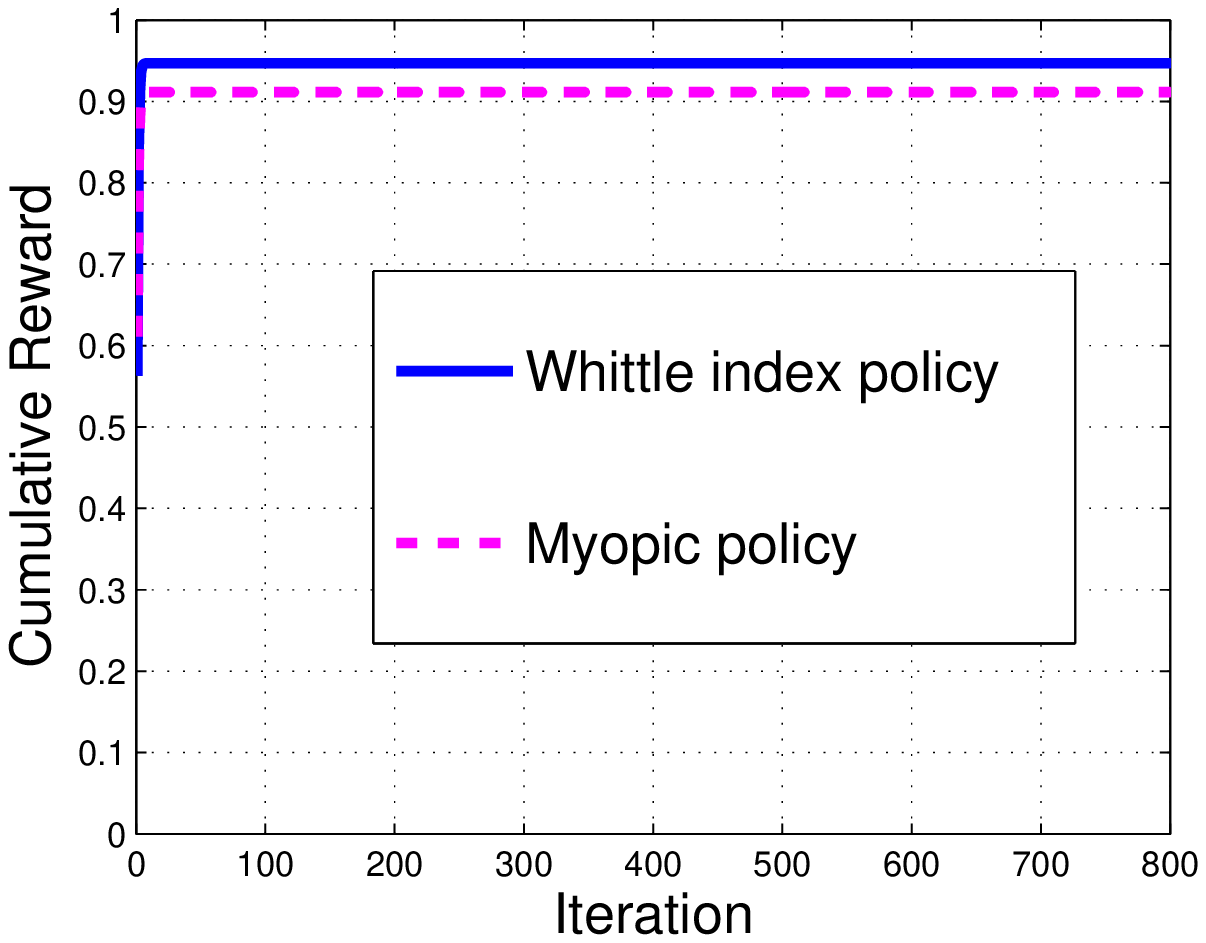}
  \label{Fig-CR-3-a} 
  }
 \\
\subfloat[]{ \includegraphics[width=0.6\columnwidth]{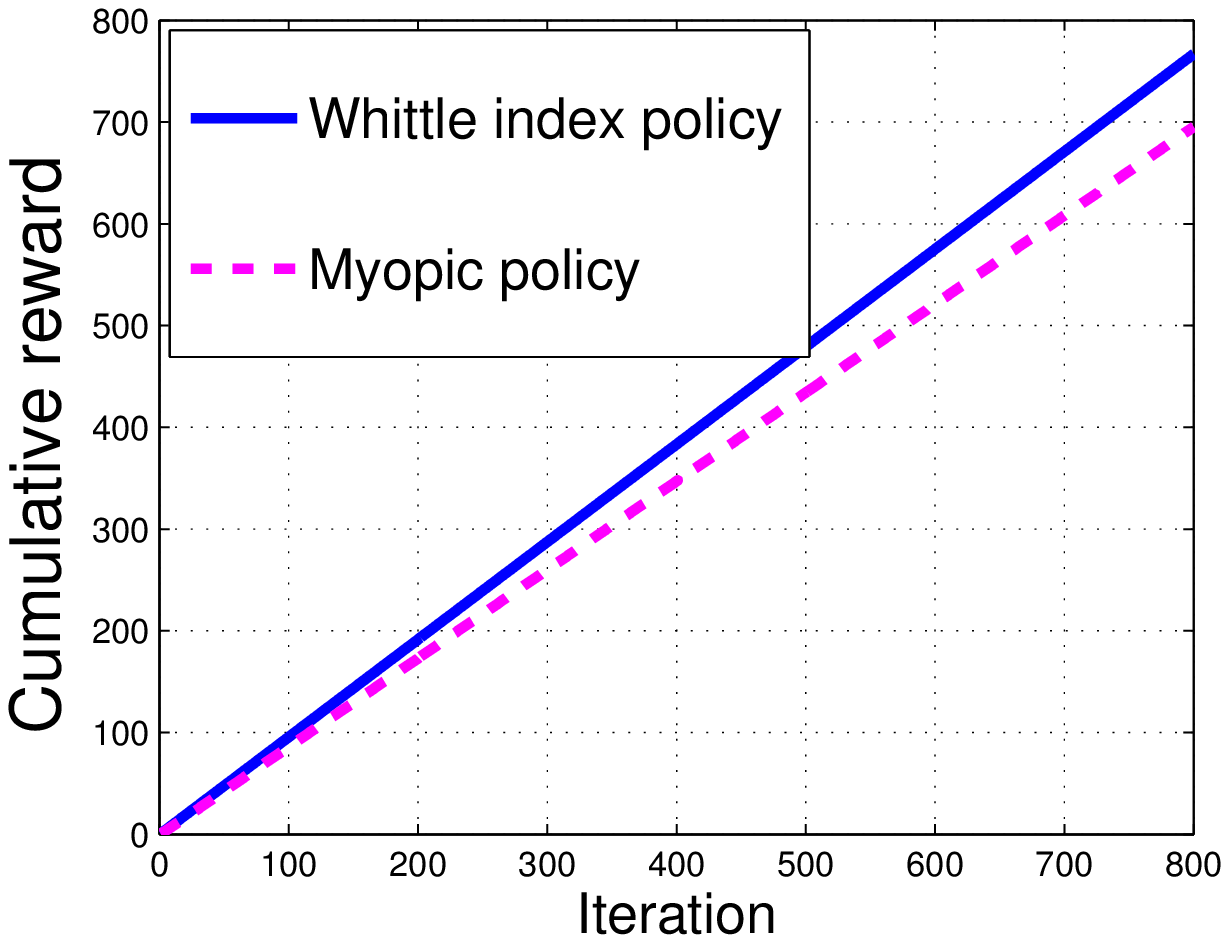}
  \label{Fig-CR-3-b}
  } 
  \caption{Cumulative reward vs time for $\beta =0.3,1$ and $N =200$  with APCS
    having both types of arms.  Fig.~\ref{Fig-CR-3-a} $\beta = 0.3$ and 
    Fig.~\ref{Fig-CR-3-b} $\beta =1.$ }
  \label{Fig-CR-3}
\end{figure*}

In this section we present some numerical results to illustrate the
performance of the Whittle index based recommendation algorithm. The
simulations use Algorithm~\ref{algo:Whittle-index-typical-viral}. We
compare the performance of Whittle index based algorithm against that
of a myopic algorithm that plays the arm that has the highest expected
reward in the step. We consider small size ($N=5,10$), medium size
($N=20$), and large size ($N=100$) systems. For all the cases we use
$\beta =0.99.$

In Fig.~\ref{Fig-CR-1}, we present numerical examples when there are a
small number of arms, i.e., $N=5.$ In this case, arms 1--4 are of type
$A$ and arm $5$ is of type $B$ arm.  The system is simulated for
parameter sets which are also shown in the figure.  In all the cases,
the initial belief used is $\pi = [0.4, 0.4, 0.4, 0.4, 0.4].$ In the
first system $\rho_0$ and $\rho_1$ for the type $B$ arm are close to
one and both policies almost always choose that arm. Hence their
performances are also comparable. This is seen in
Fig~\ref{Fig-CR-1}(a). The behavior is similar even when the $\rho_1$s
of all the arms are comparable as in the second system with
performance shown in Fig.~\ref{Fig-CR-1}(b). In this system, in the
800 plays, the type $B$ arm was played 28 and 75 times in the Whittle
index and the myopic systems respectively. In the third small system,
the Whittle index system plays the type $B$ arms significantly more
frequently than the myopic system and has a significantly better
performance; this is shown in Fig.~\ref{Fig-CR-1}c.

Fig.~\ref{Fig-CR-2} shows the performance of the two systems for
larger systems. These are obtained as follows. For $N=10,$ we have
nine type $A$ arms and one type $B$ arm. For $N =50,$ we use $48$ type
$A$ arms and $2$ type $B$ arms. The system with $N =200,$ has $190$ type
$A$ arms and $10$ type $B$ arms. We generate reward and transition
probabilities randomly using the formula $\rho_0 = 0.01 +
0.19*\mathrm{rand}(1,N),$ $\rho_1 = 0.6 + 0.3*\mathrm{rand}(1,N),$ $p
= 0.01 +0.29*\mathrm{rand}(1,N).$ The initial belief $\pi =
0.4*\mathrm{ones}(1,N).$ We observe that the Whittle index algorithm
some gain the over myopic algorithm but the gain decreases with
increasing $N.$ The decrease is due because with large $N,$ the
waiting time for each item is large and this causes many arms to be in
state $1$ with high probability. 

In Fig.~\ref{Fig-CR-3}, we compare the performance of two systems for 
 different values of discount parameter $\beta = 0.3,1$ and $N =200.$ 
  We notice that even for small $\beta =0.3,$ 
 the Whittle index algorithm gains over myopic algorithm. 

\section{Variants of type A and type B arms}
\label{subsec:discuss}
Here, we mention few extensions of a hidden RMAB.  
By considering different structure on transition probabilities, 
we can obtain different type of arms and models. 
\begin{itemize}
\item 
We consider  few variants of type A and type B arms that
 generalized our current model. The transition probabilities for this model 
 are as follows. 
\begin{eqnarray*}
P_{01}^n(0)= 
 p_n, \ \ P_{10}^n(0) = q_n  & \  \mbox{for both type arm,}
\end{eqnarray*}
\begin{eqnarray*}
P_{01}^{n}(1) = 
\begin{cases}
\epsilon P_{01}^{n}(0) & \ \mbox{if arm is type A,} \\
(1-\epsilon)  + \epsilon P_{01}^{n}(0) & \  \mbox{if arm is type B,} \\
\end{cases} \\
P_{10}^{n}(1) = 
\begin{cases}
(1-\epsilon) + \epsilon P_{10}^n(0) & \ \mbox{if arm is type A,} \\
\epsilon P_{10}^n(0) & \  \mbox{if arm is type B,} \\
\end{cases}
\end{eqnarray*}
for $\epsilon \in [0,1).$ This may be thought as the arm evolving with 
different speeds under actions $1$ and $0.$
This is referred to  dual speed restless bandit in 
\cite[Chapter $6,$ Secion $6.2$]{Gittins11}.
\item Notice that for $\epsilon = 0,$ current model  leads to a simple variant of type A and type B arm. 
For this model,  we can derive  all properties of value functions as in 
 Section~\ref{sec:whittle} using similar approach.  
Also, the Whittle index formula can be obtained. This is given in next subsection. 
\item For $\epsilon \in (0,1),$  
difficulty level of problem  increases 
significantly because  the current belief about state, $\pi_{t+1}(n)$ becomes 
non linear function of previous belief $\pi_t(n)$ when arm $n$ is played.  
Thus, it is hard to show that the arm is Whittle-indexable. But 
the approximate Whittle-indexability is proved in \cite{Meshram16a} 
under restriction on discount parameter $\beta.$ 
\end{itemize}
%
\subsection{A simple variant of the current model}
In this, we suppose $\epsilon = 0,$ and derive the value function expressions, and  Whittle index expression. To obtain these, we consider a single arm restless bandit. In such setting, we have transition probabilities as follows. $ P_{01}(0)= p,$  $P_{10}(0) = q  $ for both type arm, and   
\begin{eqnarray*}
P_{01}^{n}(1) = 
\begin{cases}
0 & \ \mbox{if arm is type A,} \\
1 & \  \mbox{if arm is type B,} \\
\end{cases} \\
P_{10}^{n}(1) = 
\begin{cases}
1 & \ \mbox{if arm is type A,} \\
0 & \  \mbox{if arm is type B.} \\
\end{cases}
\end{eqnarray*}
We first provide analysis for type A arm and then for type B arm.
\begin{enumerate}
\item Type A arm: 
The dynamic programming equation for discounted reward system  is 
\begin{eqnarray*}
V_{\beta}(\pi, \lambda) &=& \max \{V_{0,\beta}(\pi,\lambda), V_{1,\beta}(\pi,\lambda)\}, \\
V_{0,\beta}(\pi,\lambda) &=& \lambda + \beta V_{\beta} (\gamma(\pi),\lambda), \\
V_{1,\beta}(\pi,\lambda) &=& \pi \rho_0 + (1-\pi) \rho_1 + \beta V_{\beta}(1, \lambda)),
\end{eqnarray*}
where $\gamma(\pi) = \pi (1-p) + (1-\pi) q.$  
We also assume that $p + q \leq 1.$  Define $\gamma^{k}(\pi) := \gamma^{k-1}(\gamma(\pi)).$
Note that as $\lim_{k \rightarrow \infty}\gamma^{k}(\pi) = \gamma_{\infty},$
 where $\gamma_{\infty} = \frac{q}{q + p}.$ 
 Also observe that as $\pi > \gamma_{\infty},$ then $\gamma^{k}(\pi)$ is decreases to 
 $\gamma_{\infty}$ with $k$ and if $ \pi < \gamma_{\infty}$ then $\gamma^{k}(\pi)$ is increases to  $\gamma_{\infty}$ with $k.$ \\
%
Mimicking the proof technique in Section~\ref{sec:whittle}, we can show that the optimal policy is of a threshold type and arm is Whittle indexable. 
Now using threshold policy result, we can   derive
the closed form  expressions for value functions. 
\begin{eqnarray*}
K(\pi,\pi_T) := \min \{ k \geq 0: \gamma^{k}(\pi) < \pi_T\}.
\end{eqnarray*}
Then 
\begin{eqnarray*}
K(\pi, \pi_T) = \nonumber 
\begin{cases}
0 \ \ \ \ \  \  \ \ \ \ \  \ \ \ \ \ \ \ \ \ \ \ \ \mbox{if $\pi<\pi_T,$ } \\
\left \lfloor {\frac{\log(\pi_T)}{\log(\gamma(\pi))}}  \right \rfloor+1 \ \ \ \ \  \mbox{ if $\pi \geq \pi_T$.}
\end{cases}
\end{eqnarray*}
To obtain value function expressions, we consider two cases.
\begin{itemize}
\item When $\pi_T \in [0, \gamma_{\infty}),$
 $K(1, \pi_T)  = \infty$ because $\gamma^{k}(\pi)$ is decreasing to $\gamma_{\infty}$ and this will never cross a threshold $\pi_T$ for finite $k.$
Also, from Theorem~\ref{thm:single-threshold-type-A}, we can have 
\begin{eqnarray*}
V_{0,\beta}(1, \lambda) &=& \lambda + \beta V_{0,\beta} (\gamma(1), \lambda) \\
&=& \lambda + \beta \lambda + \beta^2 V_{0,\beta} (\gamma^2(1), \lambda).
\end{eqnarray*}
And $\lim_{k \rightarrow \infty}\gamma^{k}(1) = \gamma_{\infty}.$  Expanding recursion 
of $V_{0,\beta}(\gamma^{k}(1), \lambda) $ we obtain 
\begin{eqnarray*}
V_{0,\beta}(1, \lambda) &=& \frac{\lambda}{1-\beta}. \\
V_{1,\beta}(\pi, \lambda) &=& \pi \rho_0 + (1-\pi) \rho_1 + \beta \frac{\lambda}{1-\beta}.
\end{eqnarray*}
Note that $\pi \geq \pi_T,$ $\gamma^{k}(\pi) \geq  \pi_T$ for any $k \geq 1.$ Thus
\begin{eqnarray*}
V_{0,\beta}(\pi, \lambda) &=&
\begin{cases}
\lambda + \beta \bigg( \gamma(\pi) \rho_0 + (1 - \gamma(\pi)) \rho_1\bigg) + \beta^2 \frac{\lambda}{1-\beta} & \mbox{if $\gamma(\pi) < \pi_T,$} \\
\frac{\lambda}{1-\beta} & \mbox{if $\gamma(\pi) \geq  \pi_T.$} 
\end{cases} 
\end{eqnarray*}
\item When $\pi_T \in [\gamma_{\infty}, 1],$ we have $K(1, \pi_T)  < \infty.$ 
Using a threshold policy result and after simplification we get 
\begin{eqnarray*}
V_{0,\beta}(1,\lambda) &=& \frac{\lambda \left( 1- \beta^{K(1, 
\pi_T)}\right) }{ \left( 1 - \beta^{(K(1, \pi_T) +1)} \right) (1-\beta) }
+ 
\beta^{K(1, \pi_T)}
\frac{\rho(\gamma^{K(1,\pi_T)} (1) )}{\left(1- \beta^{(K(1,\pi_T)+ 1)} \right)}, \\
V_{0,\beta}(\pi,\lambda) &=& \frac{\lambda \left( 1-\beta^{K(\pi, \pi_T)}
\right)}{(1-\beta)} + 
\beta^{K(\pi, \pi_T)} \rho\left(\gamma^{K(\pi,\pi_T)}(\pi) \right) 
+
\beta^{\left(K(\pi,\pi_T) +1\right)} V_{0,\beta}(1,\lambda) \\
V_{1, \beta}(\pi,\lambda) &=& \rho(\pi) + \beta V_{0,\beta} (1, \lambda).
\end{eqnarray*}
where $\rho(\pi) = \pi \rho_0 + (1-\pi) \rho_1.$
\end{itemize} 
We now derive expressions for the Whittle index.
When $\pi \in [0,\gamma_{\infty}),$ the Whittle index is 
\begin{equation*}
W(\pi) = \pi \rho_0 + (1-\pi) \rho_1.
\end{equation*}
When $\pi \in [\gamma_{\infty}, 1]$ the Whittle index is 
{\small{
\begin{eqnarray*}
W(\pi) = \rho_1 + \beta^{\left( K(1,\pi) + 1 \right)} 
(\rho_0 -\rho_1)  \gamma^{K(1,\pi)}(1) + 
\frac{(\rho_0 - \rho_1)}{(1-\beta)} 
\left(1-\beta^{(K(1,\pi) +1)}\right) \left(\pi - \beta \gamma(\pi) \right). 
\end{eqnarray*}
}}
Using the vanishing discounted approach, we can analyse average reward problem and for that we can obtain the Whittle index expression by letting discount parameter $\beta$  approach $1.$ Hence
\begin{eqnarray*}
W(\pi) &=&
\begin{cases} 
\pi \rho_0 + (1-\pi) \rho_1 & \mbox{if $\pi \in [0, \gamma_{\infty})$} \\
 \rho_1 + (\rho_0 -\rho_1)  \gamma^{K}(1) +
K (\rho_0 - \rho_1) \left(\pi - \beta \gamma(\pi) \right)
& \mbox{if $\pi \in [\gamma_{\infty},1] $}
\end{cases}
\end{eqnarray*}
where $K = K(1, \pi).$ \\
\item Type B arm: The dynamic programming equation for discounted reward  is as follows. 
\begin{eqnarray*}
V_{\beta}(\pi,\lambda) &=& \max\{ V_{1,\beta}(\pi,\lambda), V_{0,\beta}(\pi,\lambda)\}, \\ 
V_{1,\beta}(\pi,\lambda) &=& \pi \rho_0 + (1-\pi) \rho_1 + \beta V_{\beta}(0, \lambda), \\
V_{0,\beta}(\pi,\lambda) &=& \lambda + \beta V_{\beta}(\gamma(\pi), \lambda).
\end{eqnarray*}
We now obtain the value function expressions. 
For $\pi_T \in (0, 1],$ we can get  
\begin{eqnarray*}
V_{1,\beta}(0,\lambda) &=& \rho_1 + \beta V_{1,\beta}(0,\lambda) 
\end{eqnarray*}
Hence after simplification we have  
\begin{eqnarray*}
V_{1,\beta}(0,\lambda) &=& \frac{\rho_1}{1-\beta},\\
V_{1, \beta}(\pi,\lambda) 
&=& \pi \rho_0 + (1-\pi) \rho_1 + \beta  \frac{\rho_1}{1-\beta}.
\end{eqnarray*}
 If $\pi_T \in (0,\gamma_{\infty})$ then
 \begin{eqnarray*}
 V_{0,\beta}(\pi,\lambda) = 
 \begin{cases}
 \lambda + \beta \rho(\gamma(\pi)) + \beta^2 \frac{\rho_1}{1-\beta} 
 & \mbox{if $\pi < \pi_T,$} \\
 \frac{\lambda}{1-\beta} & \mbox{if $\pi \geq \pi_T.$}
 \end{cases}
\end{eqnarray*}  
If $\pi_T \in [\gamma_{\infty},1],$ then 
 \begin{eqnarray*}
 V_{0,\beta}(\pi,\lambda) = 
 \begin{cases}
\frac{\lambda(1-\beta^{K(\pi,\pi_T)})}{1-\beta} + \beta^{K(\pi,\pi_T)} 
\left( \rho(\gamma^{K(\pi,\pi_T)}(\pi)) + \beta \frac{\rho_1}{1-\beta} \right) 
 & \mbox{if $\pi < \pi_T,$} \\
\lambda + \beta \rho(\gamma(\pi)) + \beta^2 \frac{\lambda}{1-\beta}  
  & \mbox{if $\pi > \pi_T,$} \\
\frac{\lambda}{1-\beta} 
& \mbox{if $\pi = \pi_T.$}  
 \end{cases}
\end{eqnarray*}  
For $\pi_T =0$ we can obtain 
\begin{eqnarray*}
V_{1,\beta}(\pi,\lambda) &=& \pi \rho_0 + (1-\pi) \rho_1 + \beta \frac{\lambda}{1-\beta}, \\
V_{0,\beta}(\pi,\lambda) &= & \frac{\lambda}{1-\beta}.
\end{eqnarray*}
The Whittle index expression is given as.
\begin{eqnarray*}
W(\pi) &=&
\begin{cases}
 \rho_1 & \mbox{for $\pi = 0,$} \\
 (1-\beta) (\pi \rho_0 + (1-\pi) \rho_1) + \beta \rho_1
 & \mbox{for $\pi \in (0, 1].$} 
\end{cases}
\end{eqnarray*}
\end{enumerate}
For average reward, the Whittle index in this model is same as \eqref{eq:whittle-index-B-avg}.  

\section{Thompson-Sampling Based Learning}
\label{sec:learning}

The key to a useful use of the model from the preceding sections is 
the knowledge of the parameters. These are not known a 
priori in most systems. In this section, we describe an algorithm that 
learns the parameters from the available feedback. Our scheme 
is a version of Thompson sampling \cite{Thompson} which has been 
 studied for stochastic multi-armed bandits 
 \cite{AgrawalG,GopManMan14:thompson}, learning in Markov decision
processes (MDPs) \cite{Gopalan15} and in POMDPs \cite{Meshram16c}.
In fact our algorithm is an extension of the scheme for the one-armed 
bandit, modeled as a POMDP, that was described and analysed
in \cite{Meshram16c}. An important requirement of the learning algorithm 
is to have a low regret, i.e., the exploration and exploitation sequences
should be cleverly mixed to ensure that the difference between the ideal
and the realised objective functions are small. In \cite{Meshram16c} we 
formally show that the regret is logarithmic for the one-armed case. The 
algorithm for the multi-armed case is described in Algorithm~\ref{algo:TS}, 
and we expect that its performance is also good. A formal analysis is 
being worked out.  

The algorithm proceeds as follows. At the beginning, we initialize a 
prior distribution on the space of all candidate parameters models, which in 
our case is a subset $\Theta_i$ of the unit cube $[0,1]^3;$ $\Theta_i$ 
contains all possible models for the parameters $\theta_i$ of arm $i.$  
In each step, assume that the true values of the parameters are 
$\theta_i$ and use the Whittle index (or the myopic) algorithm to 
choose the arm that is to be played.  Recall that the arm with highest
index is played in the Whittle index algorithm and the arm with
highest expected reward is played in myopic algorithm.  The playing 
of the arm $A_t$ at time $t$ yields a payoff $R_t.$ This is used to update
the prior distribution $Z_i$ of the parameter space for that arm. This
update is performed using Bayes' rule and the observed reward. The
model distribution for the other arms remain unchanged. We explain the 
update mechanism next.  Let $\mathcal{B}({\Theta}_i)$ denote the
Borel $\sigma$-algebra of $\Theta_i \subset [0,1]^3$ for the arms indexed 
by $i=1,\ldots,N.$  Let $\prob{R = r \given \theta, A}$ denote the
likelihood, under the model $\theta$, of observing a reward of $r \in
\{0,1\}$ upon action $A \in \{1,\cdots,N\}.$  This likelihood can be seen to be
as follows. 
\begin{eqnarray*}
  \prob{R= r \given \theta, A = i} = 
  \begin{cases}
    f(\theta_i, i)&  \mbox{if $r=1,$} \\
    1-f(\theta_i,i)& \mbox{if $r=0.$}
  \end{cases}
\end{eqnarray*}
Here $f(\theta_i,i)$ is the probability of observing a reward of $1$
after playing arm $i$ when the parameter is $\theta_i.$ Letting $k_i$ denote
 the number of time steps since the last time that arm $i$  was played, we 
 can obtain $ f(\theta_i, i)$ as follows.
 
%
    \begin{eqnarray*}
      f(\theta_i,i) = 
      \begin{cases}
        (1-p_i)^{k_i} \rho_{i,0} + (1-(1-p_i)^{k_i}) \rho_{i,1},&\mbox{if $i$ is type A arm, } \\
        (1-(1-p_i)^{k_i}) \rho_{i,0} + (1-p_i)^{k_i} \rho_{i,1}, &\mbox{if $i$ is type B arm.}
      \end{cases}
    \end{eqnarray*}
%

This likelihood is used to update the prior distribution $Z_{i,t}$ and the 
parameters of the arm is selected from this distribution.  We reiterate that the 
parameters and the prior distribution on these parameters, of the arms 
that are not  played remain unchanged. The states of the arms, $\pi_{i,t}$ 
are now  updated and the algorithm proceeds as before.  The details are 
described in Algorithm~\ref{algo:TS}.

\begin{algorithm}[bt]
  \caption{Thompson sampling algorithm}
  \begin{algorithmic}
     \STATE {\bf Input:} 
     Set of arms $\mathcal{N} = \{ 1, 2, \ldots, N\},$  Action space 
     $\mathcal{A} = \{ 1, 2, \ldots, N\}$,
     Observation space $\mathcal{R} = \{0,1\}$, 
     \FOR {$i=1, 2, \cdots, N$}
     \STATE Parameter space $\Theta_i \subseteq  [0,1]^3$, 
     \STATE Prior probability distribution $Z_{i,0}$ over $\Theta_i$
     \STATE {\bf Sample} parameter $\theta_{i,0} \in \Theta_j$ according to $Z_{i,0}$   
     \STATE $\pi_{i,0} = 1$
     \ENDFOR
     
    \FOR { $ t = 1, 2, \ldots$ }
     \STATE $i=$ Best arm determined by Whittle index based policy using 
     $\{\theta_{t-1}, \pi_{t-1} \}$
     \STATE   {\bf Action}  $A_t = i$
     \STATE $R_t$ = Reward from action $A_t$. 
     \STATE {\bf Update} $Z_{i,t}$to
    {\small{     
     \begin{eqnarray*}
       Z_{i,t}(B_i) := \frac{\int_{B_i} \prob{R = R_t \given \;
           \theta_{t-1}, A_{t}} Z_{i,t-1}(\theta_{i})
         d\theta_i}{\int_{\Theta_i} \prob{R = R_t \given \theta_{t-1},
           A_{t}} Z_{i,t-1}(\theta_i) d\theta_i }
     \end{eqnarray*}
     }}
    \STATE {\bf Sample} parameter $\theta_{i,t} \in \Theta_i$ according to $Z_{i,t}$ 
     \FOR {$j=1, 2, \cdots, N$}
     \STATE Update $\pi_{j,t}$
     \ENDFOR
     \ENDFOR
     
  \end{algorithmic}
  \label{algo:TS}
\end{algorithm}

\subsection{Numerical results}
\label{subsec:Num-result-TS}
\begin{figure*}
\centering
  \subfloat[]{ \includegraphics[width=0.6\columnwidth]{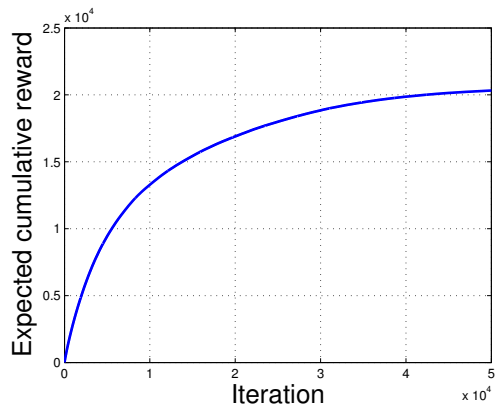}
  \label{Fig-R-a}
  }
 \\
\subfloat[]{ \includegraphics[width=0.6\columnwidth]{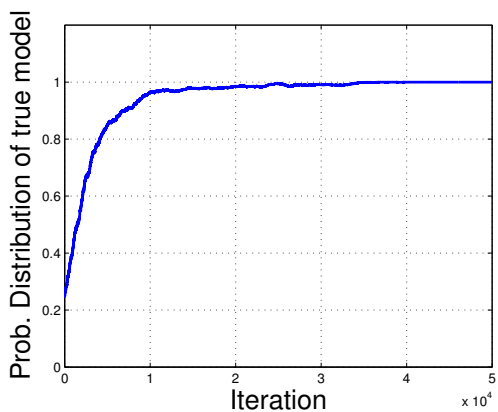}
  \label{Fig-PDR-a} } 
  \\
\subfloat[]{ \includegraphics[width=0.6\columnwidth]{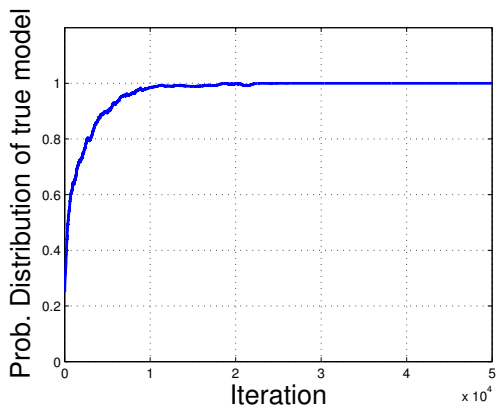}
  \label{Fig-PDR-b} } 
  \caption{Expected cumulative regret vs time horizon and probability
    mass on true model vs time horizon for both type of arms.}
  \label{Fig-TS}
\end{figure*}

We illustrate performance of Thompson sampling algorithm in
Fig.~\ref{Fig-TS} for $N =5,$ where $4$ type $A$ items and one 
type $B$ item. The true model parameters are $p = [0.15, 0.25, 0.25,
  0.15,0.15],$ $\rho_0=[0.2, 0.2, 0.1, 0.1, 0.1],$ and these are assumed 
unknown. We assume that $\rho_1 =0.7$ is known. We simulate it with
discrete parameter space into a $(2 \times 2)$ grid of $(0.15, 0.25,
0.1, 0.2).$ At the start of algorithm, we use uniform  prior over the 
$4$ points for all the arms.  We plot the expected cumulative
regret as function of time horizon, see Fig.~\ref{Fig-TS}-a. The
regret incurs whenever sample model different from true model.  In
Fig.~\ref{Fig-TS}-b and Fig.~\ref{Fig-TS}-c, we plot the probability
distribution on the true model against time for both types of items.  
We note that probability distribution of true
model is approaching to $1,$ in  Whittle index based 
algorithm. This suggests that  the Thompson sampling strategy indeed learns
the true model rather quickly. A more detailed analysis is being performed. 


\section{Conclusion}
In this paper we studied a restless multi-armed bandit for automated
playlist recommendation system with two types of items.
We considered infinite horizon discounted and average reward problem.
 We show that
both types arm are indexable and we derived the closed form expression
for the Whittle index derived from the state of the belief in the
state of the arms and from the model parameters.  Our numerical
results illustrate that the Whittle index algorithm can perform
 better than a myopic algorithm. 
We further discussed the dual speed restless bandit with hidden states and 
derived Whittle index expression for a variant.
 We have  proposed a
Thompson sampling based learning algorithm to learn the true model
parameters. Simulation results indicate that the learning is indeed
effective. The performance guarantees of the learning algorithm are
being investigated.

\bibliography{restless-bandits}

\begin{thebibliography}{29}
\providecommand{\natexlab}[1]{#1}
\providecommand{\url}[1]{\texttt{#1}}
\expandafter\ifx\csname urlstyle\endcsname\relax
  \providecommand{\doi}[1]{doi: #1}\else
  \providecommand{\doi}{doi: \begingroup \urlstyle{rm}\Url}\fi

\bibitem[Agrawal and Goyal(2012)]{AgrawalG}
S.~Agrawal and N.~Goyal.
\newblock Analysis of {T}hompson sampling for the multi-armed bandit problem.
\newblock \emph{{JMLR} {W}orkshop and Conf. Proc.}, 23:\penalty0 39.1--39.26,
  2012.

\bibitem[Auer et~al.(2002)Auer, Cesa-Bianchi, and Fischer]{AuerCF02}
P.~Auer, N.~Cesa-Bianchi, and P.~Fischer.
\newblock Finite-time analysis of the multiarmed bandit problem.
\newblock \emph{Machine Learning}, 47\penalty0 (2-3):\penalty0 235--256, 2002.

\bibitem[Avrachenkov and Borkar(2015)]{Avrachenkov15}
K.~Avrachenkov and V.~S. Borkar.
\newblock Whittle index policy for crawling ephemeral content.
\newblock Technical Report Report No.~8702, INRIA, 2015.
\newblock URL \url{https://hal.archives-ouvertes.fr/}.

\bibitem[Bertsekas(1995{\natexlab{a}})]{BertsekasV195}
D.~P. Bertsekas.
\newblock \emph{Dynamic Programming and Optimal Control}, volume~1.
\newblock Athena Scientific, Belmont, Massachusetts, 1st edition,
  1995{\natexlab{a}}.

\bibitem[Bertsekas(1995{\natexlab{b}})]{BertsekasV295}
D.~P. Bertsekas.
\newblock \emph{Dynamic Programming and Optimal Control}, volume~2.
\newblock Athena Scientific, Belmont, Massachusetts, 1st edition,
  1995{\natexlab{b}}.

\bibitem[Bubeck and Bianchi(2012)]{Bubeck12}
S.~Bubeck and N.~C. Bianchi.
\newblock Regret analysis of stochastic and non-stochastic multi-armed bandit
  problem.
\newblock \emph{Foundations and {T}rends in {M}achine Learning}, 5\penalty0
  (1):\penalty0 1--122, 2012.

\bibitem[Candes and Tao(2010)]{Candes10}
E.~Candes and T.~Tao.
\newblock The power of convex relaxation: Near optimal matrix completion.
\newblock \emph{{IEEE} {T}ransactions on {I}nformation {T}heory}, 56\penalty0
  (5):\penalty0 2053--2080, May 2010.

\bibitem[Caron et~al.(2012)Caron, Kveton, Lelarge, and Bhagat]{Caron12}
S.~Caron, B.~Kveton, M.~Lelarge, and S.~Bhagat.
\newblock Leveraging side observations in stochastic bandits.
\newblock \emph{Arxiv}, 2012.

\bibitem[Chapelle and Li(2011)]{Chapelle11}
O.~Chapelle and L.~Li.
\newblock An empirical evaluation of {T}hompson sampling.
\newblock In \emph{Proc. {NIPS}}, 2011.

\bibitem[Gittins et~al.(2011)Gittins, Glazebrook, and Weber]{Gittins11}
J.~Gittins, K.~Glazebrook, and R.~Weber.
\newblock \emph{Multi-armed Bandit Allocation Indices}.
\newblock John Wiley and Sons, New York, 2nd edition, 2011.

\bibitem[Gopalan and Mannor(2015)]{Gopalan15}
A.~Gopalan and S.~Mannor.
\newblock Thompson sampling for learning parameterized {M}arkov decision
  processes.
\newblock In \emph{Proc. {COLT}}, 2015.

\bibitem[Gopalan et~al.(2014)Gopalan, Mannor, and
  Mansour]{GopManMan14:thompson}
A.~Gopalan, S.~Mannor, and Y.~Mansour.
\newblock Thompson sampling for complex online problems.
\newblock In \emph{Proc. {ICML}}, 2014.

\bibitem[Hariri et~al.(2012)Hariri, Mobasher, and Burke]{Hariri12}
N.~Hariri, B.~Mobasher, and R.~Burke.
\newblock Context-aware music recommendation based on latent topic sequential
  patterns.
\newblock In \emph{Proc. {ACM} {RecSys}}, 2012.

\bibitem[Lai and Robbins(1985)]{Lai85}
T.~L. Lai and H.~Robbins.
\newblock Asymptotically efficient adaptive allocation rules.
\newblock \emph{Advances in Applied Mathematics}, 6\penalty0 (1):\penalty0
  4--22, March 1985.

\bibitem[Langford and Zhang(2007)]{Langford07}
J.~Langford and T.~Zhang.
\newblock The epoch-greedy algorithm for contextual multi-armed bandits.
\newblock In \emph{Proc. {NIPS}}, 2007.

\bibitem[Li et~al.(2010)Li, Chu, Langford, and Schapire]{Li10}
L.~Li, W.~Chu, J.~Langford, and R.~E. Schapire.
\newblock A contextual-bandit approach to personalized news article
  recommendation.
\newblock In \emph{Proc. {ACM} {WWW}}, 2010.

\bibitem[Liu et~al.(2013)Liu, Liu, and Zhao]{LiuZhao13}
H.~Liu, K.~Liu, and Q.~Zhao.
\newblock Learning in a changing world: {R}estless multiarmed bandit with
  unknown dynamics.
\newblock \emph{IEEE Transactions on Information Theory}, 59\penalty0
  (3):\penalty0 1902--1916, March 2013.

\bibitem[Liu and Zhao(2010)]{LiuZhao10}
K.~Liu and Q.~Zhao.
\newblock Indexability of restless bandit problems and optimality of {W}hittle
  index for dynamic multichannel access.
\newblock \emph{IEEE Transactions Information Theory}, 56\penalty0
  (11):\penalty0 5557--5567, November 2010.

\bibitem[Meshram et~al.(2015)Meshram, Manjunath, and Gopalan]{Meshram15}
R.~Meshram, D.~Manjunath, and A.~Gopalan.
\newblock A restless bandit with no observable states for recommendation
  systems and communication link scheduling.
\newblock In \emph{Proc. {IEEE CDC}}, 2015.

\bibitem[Meshram et~al.(2016{\natexlab{a}})Meshram, Gopalan, and
  Manjunath]{Meshram16b}
R.~Meshram, A.~Gopalan, and D.~Manjunath.
\newblock Optimal recommendation to users that react: Online learning for a
  class of {POMDP}s.
\newblock In \emph{Proc. {IEEE CDC}}, 2016{\natexlab{a}}.

\bibitem[Meshram et~al.(2016{\natexlab{b}})Meshram, Gopalan, and
  Manjunath]{Meshram16c}
R.~Meshram, A.~Gopalan, and D.~Manjunath.
\newblock Optimal recommendation to users that react: Online learning for a
  class of {POMDP}s.
\newblock {A}rxiv, 2016{\natexlab{b}}.

\bibitem[Meshram et~al.(2016{\natexlab{c}})Meshram, Manjunath, and
  Gopalan]{Meshram16a}
R.~Meshram, D.~Manjunath, and A.~Gopalan.
\newblock On the {W}hittle index for restless multi-armed hidden markov
  bandits.
\newblock {A}rxiv, 2016{\natexlab{c}}.

\bibitem[Meshram et~al.(2017)Meshram, Gopalan, and Manjunath]{Meshram17}
R.~Meshram, A.~Gopalan, and D.~Manjunath.
\newblock Restless bandits that hide their hand and recommendation systems.
\newblock In \emph{Proc. {IEEE COMSNETS}}, 2017.

\bibitem[Papadimitriou and Tsitsiklis(1999)]{Papadimitriou99}
C.~H. Papadimitriou and J.~H. Tsitsiklis.
\newblock The complexity of optimal queueing network control.
\newblock \emph{Mathematics of Operations Research}, 24\penalty0 (2):\penalty0
  293--305, May 1999.

\bibitem[Ross(1971)]{Ross71}
S.~M. Ross.
\newblock Quality control under {M}arkovian deterioration.
\newblock \emph{Management Science}, 17\penalty0 (9):\penalty0 587--596, May
  1971.

\bibitem[Ross(1993)]{Ross93}
S.~M. Ross.
\newblock \emph{Applied Probability Models with Optimization Applications}.
\newblock Dover Publications, 1993.

\bibitem[Thompson(1933)]{Thompson}
W.~R. Thompson.
\newblock On the likelihood that one unknown probability exceeds another in
  view of the evidence of two samples.
\newblock \emph{Biometrika}, 24\penalty0 (3--4):\penalty0 285--294, 1933.

\bibitem[Walter(1976)]{rudin-principles}
R.~Walter.
\newblock \emph{Principles of Mathematical Analysis}.
\newblock McGraw-Hill Book Co., {T}hird edition, 1976.

\bibitem[Whittle(1988)]{Whittle88}
P.~Whittle.
\newblock Restless bandits: {A}ctivity allocation in a changing world.
\newblock \emph{Journal of Applied Probability}, 25\penalty0 (A):\penalty0
  287--298, 1988.

\end{thebibliography}

\appendix
\section{Appendix}
\subsection{Proof of Lemma \ref{lemma:struct-results-2}}
\label{proof:struct-results-2}
The proof is similar for type A and type B arm. It has minor variations due to
value function expressions. Here, we present the proof for type A arm. We omit the 
proof for type B arm. The proof is using induction techniques.

\begin{enumerate}
\item Let
  \begin{eqnarray}
    V_{\beta,1}(\pi,\lambda) &=& \max \{  \lambda, \pi \rho_0 +(1-\pi) \rho_1 \}
    \nonumber\\
    V_{\beta,n+1}(\pi,\lambda)  &=&  \max \left\{ \lambda + \beta
    V_{\beta,n}((1-p)\pi, \lambda), \right. \nonumber \\ & &\left. 
    \pi \rho_0 +(1-\pi) \rho_1 +   
    \beta V_{\beta,n}(1, \lambda) 
    \right\}  \nonumber \\
    \label{algo:iterAlgo_diff}
  \end{eqnarray}
The partial derivative of $V_{\beta,1}(\pi,\lambda)$ w.r.t. $\pi$ is
$0$ or $-(\rho_1-\rho_0),$ depending on $\pi,$ and $\lambda.$ Thus the
absolute value of slope of $V_{\beta,1}(\pi,\lambda)$ w.r.t. $\pi$ is
bounded above by $(\rho_1-\rho_0).$
Making the induction hypothesis that the absolute value of slope of
$V_{\beta,n}(\pi,\lambda)$ w.r.t. $\pi$ is bounded above by
$(\rho_1-\rho_0).$
We next want to show that the absolute value of slope of $V_{\beta,n+1}(\pi,\lambda)$ w.r.t. $\pi$ is bounded above by $(\rho_1-\rho_0).$

Note that derivative of the term $\lambda + \beta
    V_{\beta,n}((1-p)\pi, \lambda)$ w.r.t. $\pi$ is bounded by
$(\rho_0-\rho_1)$ because first term is constant and second term's 
derivative is bounded by $\beta (1-p) (\rho_0 - \rho_1),$ this is bounded 
by $(\rho_1 - \rho_0).$ 

Also, the absolute value of slope of $\pi \rho_0 +(1-\pi) \rho_1 +   
   \beta V_{\beta,n}(1, \lambda) $ w.r.t. $\pi$ is bounded 
   $(\rho_1 - \rho_0)$ because first term's slope is $(\rho_0 - \rho_1)$ 
   and second term is constant. 
 Hence the absolute value of slope of   $V_{\beta,n+1}(\pi,\lambda)$ w.r.t. $\pi$ is bounded above by $(\rho_1-\rho_0).$

 By induction, it is true
for all $n \geq 1.$ From \cite[Chapter $7$]{BertsekasV195},
\cite[Proposition $2.1,$ Chapter $2$]{BertsekasV295},
$V_{\beta,n}(\pi, \lambda) \to V_{\beta}(\pi, \lambda),$ uniformly.
Thus the absolute value of slope of $V_{\beta}(\pi,\lambda)$ w.r.t. $\pi$ is
bounded above by $(\rho_1-\rho_0).$

\item The partial derivative of $V_{\beta,1}(\pi,\lambda)$ in
  \eqref{algo:iterAlgo_diff} w.r.t.  $\lambda$ is $1$ or $0,$
  depending on $\pi,$ and $\lambda.$ Thus $\frac{\partial
    V_{\beta,1}(\pi,\lambda) }{\partial \lambda}< \frac{1}{1-\beta}$
  for $0<\beta <1.$ By induction hypothesis $\frac{\partial
    V_{\beta,n}(\pi,\lambda)}{\partial \lambda} < \frac{1}{1-\beta}.$
  The partial derivative of first term in \eqref{algo:iterAlgo_diff}
  w.r.t. $\lambda$ is
  \begin{eqnarray*}
    1 + \beta \frac{\partial V_{\beta,n}((1-p)\pi,\lambda)}{\partial
      \lambda}.
  \end{eqnarray*}
  It is bounded above by $ \frac{1}{1-\beta}$ by our assumption.  The
  partial derivative of second term in \eqref{algo:iterAlgo_diff}
  w.r.t.  $\lambda$ is
  \begin{eqnarray*}
    \beta \frac{\partial V_{\beta,n}(1,\lambda)}{\partial \lambda}
  \end{eqnarray*}
  It is also bounded above by $\frac{1}{1-\beta}.$ Hence the partial
  derivative of $V_{\beta,n+1}(\pi,\lambda) $ w.r.t. $\lambda$ is
  bounded above by $\frac{1}{1-\beta}.$ By induction, it is true for
  all $n \geq 1.$ Using earlier technique, $V_{\beta,n}(\pi, \lambda)
  \to V_{\beta}(\pi, \lambda),$ uniformly. Therefore, $\frac{\partial
    V_{\beta}(\pi,\lambda) }{\partial \lambda} < \frac{1}{1-\beta}.$

\end{enumerate}
This completes the proof.
\qed

\subsection{Proof of Lemma~\ref{lemma:v-pi-monotone}}
\label{proof:v-pi-monotone}
The proof is analogous for both type A and type B arm. Also, it lead to same 
Lipschitz constant.  Here, we detail the proof for only type A arm and omit it
 for type B arm. 

 Fix $\lambda, \beta.$ Define
  \begin{eqnarray*}
    d_{\lambda}(\pi) := V_{1,\beta}(\pi,\lambda) - V_{0,\beta}(\pi,\lambda)
  \end{eqnarray*}
  We have to show that $d_{\lambda}(\pi_1 ) > d_{\lambda}(\pi_2) $
  whenever 
  $\pi_2 > \pi_1,$ for all $\pi_1, \pi_2 \in[0,1].$ Now
  %
      \begin{eqnarray*}
        d_{\lambda}(\pi_2) - d_{\lambda}(\pi_1 ) =  \beta \left(V_{\beta}((1-p)\pi_1,
        \lambda) - V_{\beta}((1-p)\pi_2, \lambda)\right) 
        - (\rho_1 - \rho_0)(\pi_2 - \pi_1).
      \end{eqnarray*}
  %
  From Lemma~\ref{lemma:struct-results-2}-$1$, we obtain
  \begin{eqnarray*}
    V_{\beta}((1-p)\pi_1, \lambda) - V_{\beta}((1-p)\pi_2, \lambda) <
    (1-p) 
    (\rho_1-\rho_0)|\pi_1 - \pi_2|.
  \end{eqnarray*} 
  Moreover,
  \begin{eqnarray*}
    \beta (1-p)(\rho_1-\rho_0)|\pi_1 - \pi_2| < (\rho_1-\rho_0)|\pi_1
    - \pi_2|.
  \end{eqnarray*}
  This implies $d_{\lambda}(\pi_2) - d_{\lambda}(\pi_1 ) < 0$ and our claim follows.
\qed

\subsection{Proof of Lemma~\ref{lemma:struct-result-3}}
\label{proof:struct-results-3}
\begin{itemize}
\item Type A arm:

 Fix $\pi,\beta.$ It is enough to show that 
 $\frac{\partial d_{\lambda}(\pi)}{\partial \lambda} =  \frac{\partial V_{1,\beta}(\pi, \lambda)}{\partial \lambda} - 
 \frac{\partial V_{0,\beta}(\pi, \lambda)}{\partial \lambda} < 0.$
From equation~\eqref{eq:value-fun-type-a}, taking partial derivative w.r.t. $\lambda,$ 
we obtain
\begin{eqnarray*}
\frac{\partial V_{1,\beta}(\pi, \lambda)}{\partial \lambda} &=& \beta 
\frac{\partial V_{0,\beta}(1, \lambda)}{\partial \lambda}, 
 \\
\frac{\partial V_{0,\beta}(\pi, \lambda)}{\partial \lambda} &=& \frac{1-\beta^{K(\pi, \pi_T)}}{1-\beta}+ \beta^{(K(\pi, \pi_T)+1)}  \frac{\partial V_{0,\beta}(1, \lambda)}{\partial \lambda}.
\end{eqnarray*}
Then
{\small{
\begin{eqnarray*}
\frac{\partial V_{1,\beta}(\pi, \lambda)}{\partial \lambda} - 
\frac{\partial V_{0,\beta}(\pi, \lambda)}{\partial \lambda} &=&
\beta \frac{\partial V_{0,\beta}(1, \lambda)}{\partial \lambda} (1 - \beta^{K(\pi,\pi_T)}) 
 - \frac{(1-\beta^{K(\pi,\pi_T)})}{(1-\beta)}. 
\end{eqnarray*}
}}
Rewriting, we have
\begin{eqnarray*}
\frac{\partial d_{\lambda}(\pi)}{\partial \lambda} &=&
(1-\beta^{K(\pi,\pi_T)})
\left[\beta \frac{\partial V_{0,\beta}(1, \lambda)}{\partial \lambda} 
- \frac{1}{(1-\beta)} \right].
\end{eqnarray*}
From~\eqref{eq:value-fun-type-a}, we can obtain 
\begin{eqnarray*}
\frac{\partial V_{0,\beta}(\pi,\lambda)}{\partial\lambda } 
= \frac{(1-\beta^{\widetilde{K}})}{(1-\beta^{\widetilde{K} +1})(1-\beta)},
\end{eqnarray*}
where $\widetilde{K} = K(1, \pi_T).$
After substitution and simplifying expressions, we have
\begin{eqnarray*}
\frac{\partial d_{\lambda}(\pi)}{\partial \lambda} =
\frac{(1-\beta^{K(\pi,\pi_T)})}{(1-\beta)} \frac{\left( \beta - \beta^{\widetilde{K}+1} - 1 +\beta^{\widetilde{K}+1} \right)}{(1- \beta^{\widetilde{K}+1})}.
\end{eqnarray*}
Clearly, $\frac{\partial d_{\lambda}(\pi)}{\partial \lambda} < 0$ for $\beta \in [0,1).$ 

\item Type B arm: 

From~\eqref{eq:value-fun-type-b-piTnonz},  
 when $\pi_T \in (0, 1]$ and all $\pi \in [0,1]$ we can obtain following  
\begin{eqnarray*}
\frac{\partial V_{1,\beta}(\pi, \lambda)}{\partial \lambda} &=& 0, \\ 
\frac{\partial V_{0,\beta}(\pi, \lambda)}{\partial \lambda} &=& 
\begin{cases}
1 & \mbox{if $T(\pi) < \pi_T,$} \\
\frac{1}{1-\beta} & \mbox{if $T(\pi) \geq \pi_T.$}
\end{cases}
\end{eqnarray*} 
Clearly, we have $\frac{\partial d_{\eta}(\pi)}{\partial \eta} < 0$ for 
$\pi_T \in (0,1].$ 
When $\pi_T = 0,$ we can get
\begin{eqnarray*}
\frac{\partial V_{1,\beta}(\pi, \lambda)}{\partial \lambda} &=& \frac{\beta}{1-\beta},\\ 
\frac{\partial V_{0,\beta}(\pi, \lambda)}{\partial \lambda} &=& 
 \frac{1}{1-\beta}.
\end{eqnarray*} 
Hence $\frac{\partial d_{\lambda}(\pi)}{\partial \lambda} < 0.$
\end{itemize}
This completes the proof.
\qed


\end{document}